%% file: final.tex
\documentclass[journal]{new-aiaa}
\usepackage[utf8]{inputenc}
\usepackage{textcomp}

\usepackage{graphicx}
\usepackage{amsmath}
\usepackage[version=4]{mhchem}
\usepackage{siunitx}
\usepackage{longtable,tabularx}
\usepackage{subcaption}
\usepackage{color}

\newtheorem{theorem}{\textbf{Theorem}}

\newtheorem{lemma}{\textbf{Lemma}}

\newtheorem{proposition}{\textbf{Proposition}}
\newtheorem{remark}{\textbf{Remark}}

\title{Collaborative target-tracking control using multiple autonomous fixed-wing UAVs with constant speeds}

\author{Zhiyong Sun \footnote{Assistant professor at Department of Electrical Engineering, Eindhoven University of Technology, Eindhoven, the Netherlands. Email:
{\tt\small sun.zhiyong.cn@gmail.com}}}
\affil{Department of Electrical Engineering, Eindhoven University of Technology,
 5600 MB Eindhoven, the Netherlands}
\author{H\'{e}ctor Garcia de Marina\footnote{Researcher in the Department of Computer Architecture and Automatic Control at the Faculty of Physics, Universidad Complutense de Madrid. Email: {\tt\small hgarciad@ucm.es}}}
\affil{Universidad Complutense de Madrid, 28040 Madrid, Spain}
\author{Brian D. O. Anderson\footnote{Emeritus professor at Australian National University, Canberra, Australia.}}
\affil{Research School of Electrical, Energy and Material Engineering, Australian National University, Canberra, ACT 2601, Australia, and   Data61-CSIRO, Canberra, ACT 2601, Australia}
\author{Changbin Yu\footnote{Professor at Curtin University, Perth, WA, Australia. }}
\affil{Optus-Curtin Centre of Excellence in Artificial Intelligence, Curtin University, Perth, WA, Australia}

\begin{document}

\maketitle

\begin{abstract}
This  paper considers a collaborative tracking control problem using a group of fixed-wing unmanned aerial vehicles (UAVs) with constant and non-identical speeds. The dynamics of fixed-wing UAVs are modelled by unicycle-type equations with nonholonomic constraints,   assuming that UAVs fly at constant altitudes in the nominal operation mode. The controller is designed such that all fixed-wing UAVs as a group can collaboratively track a desired target's position and velocity. We first present conditions on the relative speeds of tracking UAVs and the target to ensure that the tracking objective can be achieved when UAVs are subject to constant speed constraints. We construct a reference velocity that includes both the target's velocity and position as feedback, which is to be tracked by the group centroid. In this way, all vehicles' headings are controlled such that the group centroid follows a reference  trajectory that successfully tracks the target's trajectory. A spacing controller is further devised to ensure that all vehicles stay close to the group centroid trajectory. Trade-offs in the controller design and performance limitations of the target tracking control due to the constant-speed constraint are also discussed in detail. Experimental results with three fixed-wing UAVs tracking a target rotorcraft are provided.
\end{abstract}




\section{Introduction}
\subsection{Background and motivation}
Large-scale operations involving search and rescue, disaster response, environmental monitoring and sports coverage are envisioned to be more cost effective by making full use of networked multi-vehicle systems. One of the most active and important challenges in multi-vehicle systems is the control and coordination of a group of Unmanned Aerial Vehicles (UAVs) \cite{manathara2011multiple,anderson2008uav}, in particular, fixed-wing aircraft. A particular constraint complicating cooperative control design arises from constraints on their airspeeds.
In practice, airspeeds for fixed-wing UAVs should lie in a bounded interval: a lower bound for  the UAV speed that guarantees  they will not stall, and an upper bound arising from actuator constraints. In fact, small  fixed-wing aircraft typically fly optimally at \textit{constant} airspeeds, which are usually designated as nominal values designed for an optimal operation mode. For example, a constant speed might be given due to the optimization of the lift/drag ratio and as a consequence of having the vehicle's motor working in a nominal state with a fixed-pitch propeller. Furthermore, there often exists an optimal airspeed that is   most aerodynamically efficient for a given airframe of a fixed-wing UAV \cite{beard2012small}. Such a speed constraint imposes additional challenges for coordination control of multiple fixed-wing UAVs. 

Tracking control of stationary or moving targets by multiple fixed-wing UAVs has been a benchmark control problem in the field of multi-vehicle coordination control, which has found numerous applications in practice including target localization, surveillance, and target orbiting \cite{kang2009linear,   quintero2010optimal,   regina2011uav, ren2004trajectory, regina2013surface, zhang2018multiple, anjaly2019target, he2019trajectory, shaferman2008unmanned}. The footprints in Fig.~\ref{fig: tour} show four manned aircraft and helicopters tracking cyclists in the Tour de France 2018. In many stages the cyclists were split in many different groups, making it almost impossible to track all of them at the same time with only four vehicles. The usage of coordinated UAVs might solve this problem where efficient aircraft must fly at their nominal air-speeds to cover a cycling tour stage that might last several hours.  Fig.~\ref{fig:US_UAV_test} shows  two other typical scenarios involving multi-UAVs and target tracking, performed by the US military in a Perdix UAV swarm demonstration\footnote{Images in Fig.~\ref{fig:US_UAV_test} are captured from the video released in https://www.dvidshub.net/video/504622/perdix-swarm-demo-oct-2016. Also see the news report http://www.bbc.com/news/technology-38569027, dated on 10 January 2017.}. The demonstration employed almost 50 UAVs with adjustable cruising speeds and heading velocity, in an apparently centralized control framework. 

These examples motivate us to study and explore new approaches for the coordination control of a team of autonomous unmanned fixed-wing aircraft to perform a collaborative target-tracking task.  By assuming that each vehicle flies at a constant altitude in a nominal operation mode, the dynamics of fixed-wing UAVs can be modelled by 2-D differential equations with nonholonomic motion constraints and constant speeds. This paper focuses on the design of  feasible  target-tracking controllers for multiple \textit{autonomous} fixed-wing UAVs with motion and speed constraints to cooperatively track a moving target. 

\begin{figure}
  \centering
  \includegraphics[width=0.6\columnwidth]{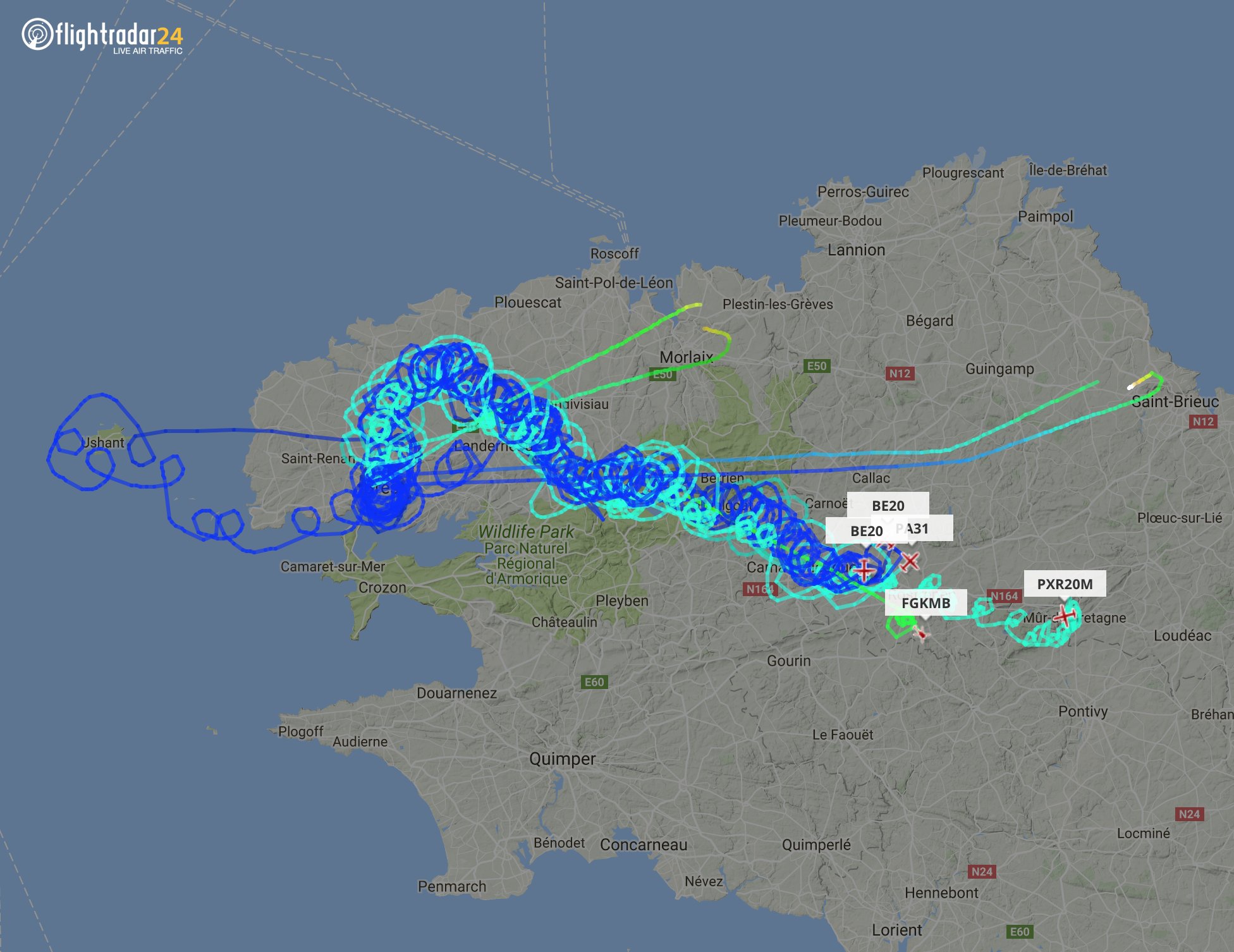}
  \caption{Footprints of the aircraft and helicopters tracking and covering the 6th stage of Tour de France 2018. Courtesy of www.flightradar24.com.}
  \label{fig: tour}
\end{figure}

\begin{figure}
  \centering
  \includegraphics[width=0.7\columnwidth]{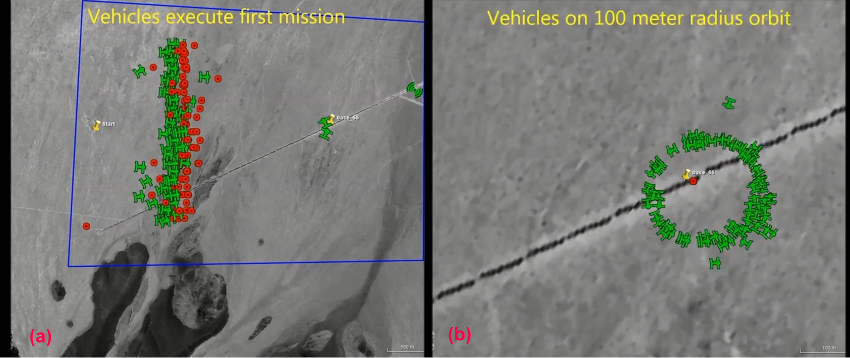}
  \caption{Typical tracking task involving multiple UAVs in a swarm. (a) Target tracking; (b) target orbiting. The images are taken from the Perdix micro-UAV swarm demonstration, released by the US Office of the Secretary of Defense Public Affairs in a public website www.dvidshub.net.    }
  \label{fig:US_UAV_test}
\end{figure}

\subsection{Related papers}

The above-mentioned coordination problems become much more challenging if all UAVs in the group have speed constraints. In fact, a more realistic model  than single- or double-integrators that can describe the nonholonomic constraints of such fixed-wing UAV dynamics is the unicycle model. Early contributions on coordination control of unicycle-type vehicles include consensus-based formation control \cite{lin2005necessary}, pursuit formation design \cite{marshall2006pursuit}, and rendezvous control \cite{dimarogonas2007rendezvous}. Other recent papers include different control constraints \cite{mastellone2008formation, quintero2010optimal,   regina2011uav,  regina2012novel, liu2013distributed, fidan2013single, Arranz2014TAC} to name just a few, but all assume that both the cruising speed and heading angular speed of individual vehicles are adjustable or controllable. For example, collaborative target-tracking guidance with fixed-wing UAVs was discussed in e.g., \cite{kang2009linear,   quintero2010optimal,   regina2011uav,  regina2012novel} via several strategies such as \textit{model predictive control} or \textit{dynamic programing}.  The tracking control of multiple unicycles was considered in \cite{Arranz2014TAC, arranz2011cooperative}, in which a  group of unicycles were tasked to track the trajectory of a  target with a time-varying velocity and the framework of \emph{circular motion} control proposed  in \cite{sepulchre2007stabilization} was employed. 
A more recent paper \cite{brinon2017target} relaxed measurement requirements on target information, but still assumed a unicycle-type  model with control inputs relating to  both cruising speed and angular speed.  In this paper, we consider the more challenging tracking control problem when only the orientation can be controlled and the speeds of all the vehicles remain fixed, so as to reflect the consideration of speed constraints in certain types of real fixed-wing UAVs \cite{anderson2008uav,regina2013surface, kang2009linear}. 

When a group of  constant-speed vehicles are involved in the control task of trajectory tracking, the problem becomes even more challenging. Two fundamental tracking problems (on tracking a straight line trajectory or a circular trajectory without the consideration of velocity matching) were discussed in \cite{moshtagh2007distributed,sepulchre2007stabilization, sepulchre2008stabilization, klein2006controlled, justh2004equilibria},  which assumed 2-D unicycle-type UAV models with \textit{unit} speeds.  
In \cite{sepulchre2007stabilization, sepulchre2008stabilization} the authors showed how to control a group of \emph{unit-speed} unicycles to achieve two behavior primitives (viz.   circular motion and   translational  motion), with switching between circling and aligned translation control. The papers \cite{anderson2008uav, van2008non} investigated non-hierarchical formation control of a group of fixed-wing UAVs with speed constraints in performing surveillance tasks, while the control law is based on a switching control that regulates the center of mass to follow a nominated (spiral) trajectory.  The proposed technique in \cite{klein2006controlled} showed a two-step design approach for designing tracking controllers that allows  the formation centroid of a group of \emph{unit-speed} unicycles to track a moving target. Such strategies were further explored in \cite{moshtagh2009vision} for vision-based flocking control of multiple autonomous vehicles. 

In this paper, we further consider a more realistic scenario where constant speeds in a multi-UAV group can be  \textit{non-identical} (but may be similar in terms of their nominal values). This is motivated by practical tracking scenarios that a multi-vehicle group may consist of multiple \textit{heterogeneous} UAVs with different functions or payloads, which will help to perform a comprehensive target-tracking task. 
Recent efforts towards the coordination of fixed-wing aircraft with \textit{non-identical constant} speeds were presented in \cite{Seyboth2014,2017Sun_circular_formation,2017iroscircular}, which  demonstrated coordination algorithms based on circular motions, rigid formations and distributed consensus-based flying coordination in practice. In particular, the recent paper \cite{2017Sun_circular_formation} explores the possibility of circular formation control for a group of autonomous fixed-wing UAVs under constant speed constraints, and develops several control strategies applicable for autonomous UAV systems such as position-/displacement-based approach and distance-based approach. A general theory and experimental verification of target-tracking control with multiple fixed-wing UAVs with non-identical constant speeds is however lacking in the literature.

\subsection{Contributions and paper organization}
In this paper, we aim to provide a systematic method to solve the target-tracking control problem under the constraints that (a) a target with a general trajectory curve is to be tracked, and (b) different vehicles have non-identical constant speeds. The framework for designing tracking controllers in this paper is motivated in part by \cite{sepulchre2008stabilization, klein2006controlled}, but several significant extensions and novel designs are required to deal with heterogeneous vehicles with non-identical speeds to achieve a collaborative control task of tracking a general target trajectory. The controller design consists of two parts: reference velocity tracking that aims to regulate the group centroid to track a target's velocity and position, and a spacing control that ensures all vehicles stay close to the group centroid. This two-step controller design was suggested in \cite{klein2006controlled} for unit-speed unicycle agents, whereas in this paper we further develop an implementable set of tracking controllers for fixed-wing UAVs (modelled by unicycle-type equations) with constant but non-identical speeds. Furthermore, a rigorous analysis on the stability of different equilibria and the convergence of velocity tracking dynamics is presented in detail. Due to the coordination constraints arising from constant speed, trade-offs in the controller design are inevitable and we also provide a detailed analysis of the performance limitations of using multiple fixed-wing UAVs in a collaborative tracking task. We also remark that our recent paper \cite{2017Sun_circular_formation} on fixed-wing UAV coordination control focuses on circular formation control, which is a different control objective to the target tracking control discussed in this paper.



A preliminary conference version of this paper was presented in \cite{sun2015collective}. Compared with \cite{sun2015collective}, this paper presents substantial extensions and new technical contributions, which include detailed proofs for all key results that were omitted in \cite{sun2015collective},  major additions on the  convergence analysis of velocity tracking control with mobile targets, conditions on UAVs' constant speeds for feasible reference velocity tracking, and  comprehensive discussions with novel insights on the trade-offs and tracking performance limitations of using constant-speed UAVs in the target tracking task.
Furthermore, a new section devoted to experimental verifications is also included in this extended version to demonstrate the real-life satisfactory performance of the tracking controller and to facilitate the reproduction of UAV experiments for the general community. Thus, for the first time, this paper presents both theoretical solutions and experimental results on feasible tracking controllers for fixed-wing UAVs in autonomous target tracking tasks, and the findings will further advance real-life applications of fixed-wing UAVs, even under the strict conditions of constant cruising speed constraints. 
 
We organize this paper as follows. We introduce the UAV model and problem formulation in Section  \ref{sec:background}. Reference velocity tracking  that ensures the UAV group centroid successfully tracks a reference trajectory is discussed in Section \ref{sec:velocity_tracking}, which also shows dedicated proofs on the convergence analysis of velocity tracking control with mobile targets.  Section \ref{sec:spacing_tracking} proposes the design of a reference velocity and spacing controller  to ensure all vehicles move close to the group centroid. Experimental results with fixed-wing UAVs tracking a moving rotorcraft with the proposed collaborative  tracking controller are shown in Section \ref{sec:experiment}.   Finally, Section \ref{sec:conclusion} concludes this paper.

%

\section{Background, preliminary  and problem description} \label{sec:background}
\subsection{Notations} \label{sub:notations}
The notations used in this paper are fairly standard. 
The set $\mathbb{S}^1$ denotes the unit
circle and an angle $\theta_i$ is a point  $\theta_i \in \mathbb{S}^1$. The $n$-torus is the Cartesian product $\mathbb{T}^n = \mathbb{S}^1 \times \cdots \times \mathbb{S}^1$. 
For a complex number $z \in \mathbb{C}$, $\text{Re}(z)$ and $\text{Im}(z)$  denote, respectively, its real part and imaginary part, and $\bar z$ is the complex conjugate of $z$.   For $z_1, z_2 \in \mathbb{C}^n$, the real scalar product is defined by $\langle z_1, z_2 \rangle = \text{Re}(\bar z_1^T z_2)$, i.e., the real part of the standard scalar product over $\mathbb{C}^n$.  The norm of $z \in \mathbb{C}^n$ is defined as $\|z\| = \langle z,z \rangle^{1 \over 2}$. For two complex numbers $z_k  = v_k e^{i \theta_k} \in \mathbb{C}$ and $z_j  = v_j e^{i \theta_j} \in \mathbb{C}$ represented in polar form, it holds
\begin{align} 
    \langle z_k, z_j \rangle &= \text{Re}(v_k e^{-i \theta_k} \cdot v_j e^{i \theta_j}) = \text{Re}(v_k v_j e^{i (\theta_j - \theta_k)} ) \nonumber \\
    &= \text{Re}(v_k v_j \text{cos} (\theta_j - \theta_k) + i v_k v_j \text{sin}(\theta_j - \theta_k)) \nonumber \\
    &= v_k v_j \text{cos} (\theta_j - \theta_k)
\end{align}
and similarly,
\begin{align}
    \langle z_k, i z_j \rangle &= \text{Re}(v_k e^{-i \theta_k} \cdot i v_j e^{i \theta_j}) = \text{Re}(i v_k v_j e^{i (\theta_j - \theta_k)} ) \nonumber \\
    &= \text{Re}(i v_k v_j \text{cos} (\theta_j - \theta_k) - v_k v_j \text{sin}(\theta_j - \theta_k)) \nonumber \\
    &= v_k v_j \text{sin} (\theta_k - \theta_j)
\end{align}
Furthermore, the following equalities, whose simple proofs can be found in \cite{2017Sun_circular_formation}, 
will be used frequently later.
\begin{lemma} \label{lemma:basic_calculation}
For $a, b \in \mathbb{C}^n$,
\begin{align}
\left \langle a, b \langle a,b  \rangle \right \rangle  &= \left \langle b, a \langle a,b  \rangle \right \rangle = \left \langle a, b \langle b,a  \rangle \right \rangle = \langle a, b \rangle ^2  \label{eq:lemma1_eq1},
\end{align}
and
\begin{align}
- \langle i a, b \rangle &= \langle  a, i b \rangle  \label{eq:lemma1_eq2}.
\end{align}
\end{lemma}

 \subsection{Vehicle models}

In this paper, we consider a group of $n$ fixed-wing vehicles modelled by unicycle-type kinematics subject to a nonholonomic constraint and constant-speed constraint. By following 
\cite{regina2013surface,regina2012novel,kang2009linear}, the kinematic equations of fixed-wing vehicle $k$ flying in a fixed horizontal plane are described by
\begin{align} \label{eq:vehicle_model_1}
\dot x_k &= v_k \,\, \text{cos}(\theta_k) \nonumber \\
	\dot y_k &= v_k \,\, \text{sin}(\theta_k) \\
\dot \theta_k &=   u_k  \nonumber
\end{align}
where $x_k \in \mathbb{R}, y_k  \in \mathbb{R}$ are the coordinates of vehicle $k$ in the real horizontal plane and $\theta_k$ is the heading angle. The fixed-wing UAVs have  fixed cruising speeds $v_k>0$, which in general are distinct for different vehicles;  $u_k$ is the control input to be designed for steering the orientation of vehicle $k$. The equation \eqref{eq:vehicle_model_1} serves as a high-level kinematic model that captures motion constraints and vehicle dynamics  for fixed-wing UAVs flying at \textit{trim} conditions (e.g., at constant altitudes or fixed level flight) \cite{beard2012small}. 
In fact, the model (\ref{eq:vehicle_model_1}) fits fairly well into the dynamics of a small fixed-wing aircraft as we have shown in \cite{2017iroscircular,2017Sun_circular_formation}. When an aircraft flies at its nominal airspeed, both its lift and weight are balanced so that there is no  change in the vehicle's altitude. Therefore, the 3D dynamics of the aircraft can be easily decoupled to separate the planar motion parallel to the ground, i.e., the dynamics (\ref{eq:vehicle_model_1}), and the vertical motion. Consequently, we will assume that the aircraft fly at fixed (but different) altitudes in this paper, which is common in practice for nominal operation of multi-UAV groups. 
Note that when the airspeed is much higher than the windspeed we can consider   $v_k$ as the ground speed \cite{beard2012small}. 

For the convenience of analysis we also rewrite in complex notation the model (\ref{eq:vehicle_model_1}) for vehicle $k$ as
\begin{subequations}  \label{eq:vehicle_model}
\begin{align}
\dot r_k &= v_k e^{i \theta_k}    \\
\dot \theta_k &= u_k    \label{eq:vehicle_model_theta}
\end{align}
\end{subequations}
where the vector $r_k(t) = x_k(t) + i y_k(t)  : = \|r_k\| e^{i \phi_k(t)} \in \mathbb{C}$  denotes vehicle $k$'s position in the complex plane (where $\|r_k\| := \sqrt{x_k^2 + y_k^2}$ and $\phi_k := \text{arg}(r_k)$). We also define the vectors $r = [r_1, r_2, \cdots, r_n]^T \in \mathbb{C}^n$,   $\theta = [\theta_1, \theta_2, \cdots, \theta_n]^T \in \mathbb{T}^n$ and $e^{i\theta} = [e^{i\theta_1}, e^{i\theta_2}, \cdots, e^{i\theta_n}]^T \in \mathbb{C}^n$ to  represent the stacked vectors of the positions,  headings, and orientations of all the vehicles, respectively.

\subsection{Problem formulation}
\label{sec: backpro}

We consider the problem of tracking a target trajectory $r_{\text{target}}(t)$ by a team of fixed-wing UAVs with (possibly non-identical) constant speeds. A reasonable strategy in the multi-vehicle tracking control is to use the centroid of the vehicle team, denoted by $\hat r(t) = \frac{1}{n} \sum_{k=1}^n r_k(t)$, as the \textit{surrogate}  position of the whole vehicle group, which is to be controlled to match the target position $r_{\text{target}}(t)$. In fact, we will generate a velocity reference signal $\dot r_{\text{ref}}(t)$ to be tracked by the centroid in a manner previously introduced in \cite{klein2006controlled}. As will be seen later in Section~\ref{sec:spacing_tracking}, this velocity reference signal is not necessarily identical to the target's velocity (unless initially the reference point is collocated with the target's position). The construction of the reference velocity takes into account both the target's velocity and the relative position between the target and the group centroid.    The control strategy is split into two loops similarly to \cite{xargay2012time}. In a first phase discussed in Section \ref{sec:velocity_tracking}, an inner loop for each member of the UAV team controls their orientations so that the velocity of the centroid tracks $\dot r_{\text{ref}}(t)$. In a second phase discussed in Section \ref{sec:spacing_tracking}, an outer loop generates $\dot r_{\text{ref}}(t)$ using information on the target's position.


We will see that our proposed controller for coordinating constant-speed UAVs resembles the phase control problem of coupled oscillators studied in \cite{sepulchre2008stabilization,Seyboth2014}. This is due to the fact that we will control  angular velocities of the vehicles to regulate their heading orientations. An important measure for such an oscillator network is the so-called \emph{order parameter} $p_{\theta} = \frac{1}{n} \sum_{k=1}^n  e^{i\theta_k}$, which is actually the centroid velocity of a  group of unit-speed vehicles and is often used to measure the phase coherence or phase synchronization level \cite{sepulchre2007stabilization, dorfler2014synchronization}. In this paper, we will use a similar metric called \emph{average linear momentum} $\dot {\hat r} := \frac{1}{n} \sum_{k=1}^n v_k e^{i \theta_k}$  \cite{Seyboth2014}, which is actually the velocity of the group centroid point according to the definition of $\hat r$. The authors in \cite{Seyboth2014} employ this metric for the control of circular motions of unicycles with different constant speeds, and one of our proposed controllers can be seen as an extension of that work.

Apart from the velocity reference tracking control by the group centroid, we also need to provide an extra control to steer each vehicle to stay close to the group centroid. For example, if we consider a tracking task where the position of the target is stationary, one may encounter situations where some or all of the tracking vehicles travel away from the origin while their centroid remains constant at the target, a situation that is not acceptable (except in the short term). We will therefore introduce an additional term to the tracking controller to regulate \emph{the planar spacing} among the vehicles as a coherent tracking team. Consequently, all the vehicles will remain at a bounded distance from their centroid. In this work, we will not consider additional terms for collision avoidance. In particular, since one of our goals is to demonstrate the algorithm on fixed-wing aircraft, we can always suppose that they are flying at different altitudes, which  is a common assumption in the literature  \cite{klein2006controlled, Seyboth2014}. This should not be a strong restriction, in particular  if we are not aiming at really massive swarms. In summary, the overall tracking controller will be in the form
\begin{align} \label{eq:controller_combined}
u_k = u_k^{\text{velocity}} + u_k^{\text{spacing}}
\end{align}
where $u_k^{\text{velocity}}$ is the responsible term to track the reference velocity (which is constructed by feedback from a target's position and velocity) via the group centroid and the term $u_k^{\text{spacing}}$ controls  \emph{the spacing} for each individual vehicle so that they can remain with a bounded distance to the centroid. The two-step controller design method \eqref{eq:controller_combined} was also suggested in \cite{klein2006controlled} for tracking control of unit-speed unicycles, while the development of tracking controllers for heterogeneous unicycle-type agents with constant but non-identical speeds requires substantial improvements and dedicated proofs on the stability and convergence property (provided in Section~\ref{sec:velocity_tracking}). Furthermore, it is obvious that there exist certain trade-offs in the design of these two controllers since, generally speaking,  the two sub-tasks (i.e., reference velocity tracking and spacing control) are not likely to be achieved perfectly at the same time. Performance limitations arising from the trade-offs in the controller design will be discussed in more detail in Section \ref{sec:velocity_tracking} and Section \ref{sec:spacing_tracking}.



\section{Controller design phase I: Reference velocity tracking}  \label{sec:velocity_tracking}
In this section, we discuss Phase I in the controller design, i.e., how to regulate each vehicle's heading and motion with the dynamics $\dot \theta_k = u_k^{\text{velocity}}$ so that the velocity of the vehicle group's centroid achieves a desired reference velocity.  The construction of a reference velocity, the combined controller design that stabilizes the spatial error to the target's position and the spacings between individual vehicles  will be discussed in the next section.

\subsection{Conditions on constant speeds for a feasible reference velocity tracking}


%
Denote the reference velocity by $\dot r_{\text{ref}}(t) =  v_{\text{ref}}(t) e^{i \theta_{\text{ref}}(t)}$, where $v_{\text{ref}}(t)$ and $\theta_{\text{ref}}(t)$ are the  (possibly time-varying) airspeed and heading direction of the reference, respectively. If initially the reference trajectory coincides with the target trajectory, the target velocity is used as the reference velocity. Otherwise, the construction of the reference velocity should take into account both the target's position and velocity (which will be elaborated in detail in Section \ref{sec:spacing_tracking}). For a group of constant-speed vehicles, one cannot expect that an arbitrary reference velocity can be tracked by the group centroid. In this section, we will give several conditions that guarantee a feasible reference velocity tracking. 

Recall the definitions of the group centroid position $\hat r$ and velocity $\dot {\hat r}$ defined in Section \ref{sec: backpro}, whose values depend on the simultaneous headings and velocities of each individual vehicle. The maximum value of the group centroid speed, denoted as $\bar v_{\text{max}}$, can be achieved when  \emph{all the vehicles reach heading synchronization}, in which case there holds $\bar v_{\text{max}} :=\|\dot {\hat r}_{\text{max}} \|= \frac{1}{n} \sum_{k=1}^n  v_k$. However, due to the non-identical constant speeds in the group, even if the maximum group centroid speed $\bar v_{\text{max}} = \frac{1}{n} \sum_{k=1}^n  v_k$ can be achieved, the inter-vehicle distances between individual vehicle will grow larger and larger without bound because of the non-zero differences between individual speeds. Therefore, a strict condition on the individual speeds should be imposed.  

For any vehicle in the group, the minimum speed, denoted by $v_{\text{min}} : =\text{min}_{k \in \{1,2,\cdots,n\}} v_k$ should be greater than the reference speed $\|\dot { r}_{\text{ref}}(t)\|$. Otherwise, the distance between the vehicle with the smallest speed and the target will grow unboundedly and a collective target tracking cannot be achieved. Therefore, one should ensure that 
the reference speed is smaller than the minimum speed in the vehicle group (which trivially ensures that the reference speed is smaller than any constant speed of the vehicle group). Of course, the reference speed cannot exceed the maximum speed of the group centroid $\frac{1}{n} \sum_{k=1}^n  v_k$. 

Now consider the task that the group centroid is regulated to achieve a zero centroid velocity, which is not possible if there exists one vehicle in the group whose constant speed is larger than the sum of all other vehicles' speeds. In order to ensure that the group centroid can reach a zero reference speed, one should exclude this case by imposing a further condition that  $v_{\text{max}} \leq \sum_{k=1}^n  v_k-v_{\text{max}}$. In summary, the necessary conditions for a feasible reference velocity tracking are shown in the following proposition.

\begin{proposition}  \label{ass:velocity_condition}
For a feasible reference velocity tracking, the constant speeds for the vehicle group should satisfy the following conditions: 
\begin{itemize}
\item All vehicles' speeds should be greater than the reference speed, i.e., $v_{\text{min}}\geq v_{\text{ref}}$;
\item There does not exist one vehicle in the group whose   speed is larger than the sum of the speeds of all other vehicles, i.e. there should hold $v_{\text{max}} \leq  \sum_{k=1}^n  v_k-v_{\text{max}}$.
\end{itemize}   
\end{proposition}

In the following subsections, we consider velocity tracking control for the constant-speed vehicle group. We will start from the simple case of constant reference velocity and then extend the controller design result to the time-varying velocity case. We remark that the conditions in Proposition~\ref{ass:velocity_condition}  on feasible reference velocity tracking should be imposed  to interpret the main results of this section.
\subsection{Tracking a constant reference  velocity}
%
%
%
This subsection solves the control problem of regulating the formation centroid to track and match a constant reference  velocity $\dot r_{\text{ref}} =  v_{\text{ref}} e^{i \theta_{\text{ref}}}$, in which   both  $v_{\text{ref}}$ and  $ \theta_{\text{ref}}$ are constant. The controller involves collectively regulating the heading of each individual vehicle. 
This control problem can be seen as an extension of the result in \cite{sepulchre2007stabilization, moshtagh2007distributed} which discussed the control problem of regulating a group of \emph{unit-speed} vehicles to achieve flocking behavior (i.e. a translational motion along a fixed direction).  


The first main result on constant velocity tracking is stated in the following theorem.

\begin{theorem} \label{theorem:constant}
Suppose that the  constant reference velocity $v_{\text{ref}}$ and all vehicles' constant speeds $v_k$ satisfy the conditions in Proposition~\ref{ass:velocity_condition}. Consider the following steering control law for \eqref{eq:vehicle_model}
\begin{align}\label{eq:controller_constant}
u_k^{\text{velocity}}  & =   -  \gamma \left \langle \dot {\hat r} - \dot r_{\text{ref}}, \frac{1}{n} i  v_k e^{i \theta_k} \right \rangle \nonumber \\
& =  -  \gamma \left \langle \frac{1}{n} \sum_{k=1}^n v_k e^{i \theta_k} - v_{\text{ref}} e^{i \theta_{\text{ref}}}, \frac{1}{n} i  v_k e^{i \theta_k} \right \rangle
\end{align}
where $\gamma$ is a positive control gain that can be used for adjusting turning rate. Suppose that  the initial headings of all the vehicles are not aligned with the phase of the reference velocity. Then the equilibrium point for the phase dynamics \eqref{eq:vehicle_model_theta} at which $\dot {\hat r} := \frac{1}{n} \sum_{k=1}^n v_k e^{i \theta_k} = v_{\text{ref}} e^{i \theta_{\text{ref}}}$ is asymptotically stable and all other equilibria are unstable. Furthermore, the control law \eqref{eq:controller_constant} almost globally stabilizes the group centroid velocity to the desired constant reference velocity $\dot r_{\text{ref}} =  v_{\text{ref}} e^{i \theta_{\text{ref}}}$.
\end{theorem}
We remark that the second term in \eqref{eq:controller_constant}, i.e., $i  v_k e^{i \theta_k}$, rotates the complex variable  $v_k e^{i \theta_k}$ by 90 degrees. The control input can also be written using real variables with trigonometric functions, by using the formulas introduced in Section~\ref{sub:notations}.  

\begin{proof}
We first show that the system $\dot \theta = u^{\text{velocity}}$ with the above designed controller \eqref{eq:controller_constant} describes a gradient flow for the following quadratic potential
\begin{align} \label{eq:V_function1}
V(\theta) & = \frac{1}{2}\|\dot {\hat r} - \dot r_{\text{ref}}\|^2 \nonumber \\
& = \frac{1}{2}  \left \langle \frac{1}{n} \sum_{k=1}^n v_k e^{i \theta_k} - v_{\text{ref}} e^{i \theta_{\text{ref}}}, \frac{1}{n} \sum_{k=1}^n v_k e^{i \theta_k} - v_{\text{ref}} e^{i \theta_{\text{ref}}} \right \rangle
\end{align}
The gradient of $V(\theta)$ can be calculated as
\begin{align}
\frac{\partial V(\theta)}{\partial \theta_k} & = \left \langle  \frac{\partial {\dot {\hat r}}}{\partial \theta_k},  \dot {\hat r} - \dot r_{\text{ref}} \right \rangle \nonumber \\
& = \left \langle \frac{1}{n} i  v_k e^{i \theta_k},  \dot {\hat r} - \dot r_{\text{ref}} \right \rangle  =  \left \langle  \dot {\hat r} - \dot r_{\text{ref}}, \frac{1}{n} i  v_k e^{i \theta_k} \right \rangle
\end{align}
Hence the phase system \eqref{eq:vehicle_model_theta} with the designed control law \eqref{eq:controller_constant} can be written as $\dot \theta = -  \gamma  \nabla V(\theta)$, which is a gradient descent flow for the potential function  $V(\theta)$. Furthermore, $V(\theta) \geq 0$ and $V(\theta)$ is zero if and only if $\dot {\hat r} = \dot r_{\text{ref}}$. Thus, $V(\theta)$ can be used as a Lyapunov function for the convergence analysis. Due to the gradient property of the phase system \eqref{eq:vehicle_model_theta},  there exists no limit cycle  under the control $\dot \theta = u^{\text{velocity}}$ at the steady state \cite{absil2006stable}.  Furthermore, the stability analysis of different equilibria \eqref{eq:vehicle_model} can be cast as a critical point analysis of the real analytic potential $V(\theta)$. Note that the system variable is $\theta \in \mathbb{T}^n$ where $\mathbb{T}^n$  is compact and thus  the sub-level sets of $V(\theta)$ are also compact according to its definition in \eqref{eq:V_function1}. We remark that for the phase dynamics \eqref{eq:vehicle_model_theta} with the velocity tracking controller \eqref{eq:controller_constant}, the state variable is $\theta$; the position variable $r$ is not involved. \footnote{This point will be made clear in Remark \ref{remark:real_variables}, in which the control input is equivalently written in real variables that only involve $\theta$.} 

The derivative of $V(\theta)$ along the trajectory of the phase system \eqref{eq:vehicle_model} can be computed as
\begin{align}
\dot V  & = \nabla V(\theta)^T \dot \theta  = - \frac{\gamma}{n} \sum_{k=1}^n \left \langle \dot  {\hat r} - \dot r_{\text{ref}}, i  v_k e^{i \theta_k} \right \rangle ^2 \leq 0
\end{align}
By LaSalle's Invariance Principle, all solutions of  \eqref{eq:vehicle_model} with the controller \eqref{eq:controller_constant} converge to the largest invariant set contained in
\begin{align} \label{eq:set_O}
\mathcal{O}(r, \theta) &= \{(r, \theta)| \dot V = 0\} \nonumber \\
& = \{(r, \theta) | \left \langle \dot  {\hat r} - \dot r_{\text{ref}}, i  v_k e^{i \theta_k} \right \rangle = 0, \,\, k = 1,2, \cdots, n \}
\end{align}
In the following we will show the properties of different sets of critical points. Note that the Jacobian of the right-hand side of \eqref{eq:vehicle_model} with the controller \eqref{eq:controller_constant} is $-\gamma H_V$ where $H_V$ is the Hessian of $V$.  The nature of an equilibrium (of being a   minimum, a saddle point or a   maximum) can be determined by the signs of the eigenvalues of the Hessian $H_V$ at that equilibrium assuming that the Hessian is non-singular. 

From the above analysis, it is clear that the desired critical points on which $\dot  {\hat r} : = \frac{1}{n} \sum_{k=1}^n v_k e^{i \theta_k} = v_{\text{ref}} e^{i \theta_{\text{ref}}} =: \dot r_{\text{ref}}$ are   global minima of $V(\theta)$, which are asymptotically stable. We will show other equilibrium sets that correspond to $V(\theta)>0$, or equivalently equilibrium points in the set $\mathcal{O}(r, \theta)$ in \eqref{eq:set_O} with $\left \langle \dot  {\hat r} - \dot r_{\text{ref}}, i  v_k e^{i \theta_k} \right \rangle = 0$ and $\dot  {\hat r} \neq \dot r_{\text{ref}}$, are unstable. Denote the velocity tracking error as 
\begin{align} \label{eq:velocity_tracking_error}
\dot { \tilde r} & : =  \dot  {\hat r} - \dot r_{\text{ref}}  =  \frac{1}{n} \sum_{k=1}^n v_k e^{i \theta_k}  - v_{\text{ref}} e^{i \theta_{\text{ref}}} \nonumber \\
&=  \|\dot { \tilde r}\| e^{i \phi}
\end{align}
where  $\|\dot { \tilde r}\|$ is the magnitude of $\dot { \tilde r}$, and $\phi := \text{arg}(\dot { \tilde r} )$ by definition.  We call such critical points for which $\|\dot { \tilde r}\| >0$ \emph{undesired equilibria} since they do not achieve a desired reference velocity tracking.

\begin{figure*}
  \centering
  \includegraphics[width=1\columnwidth]{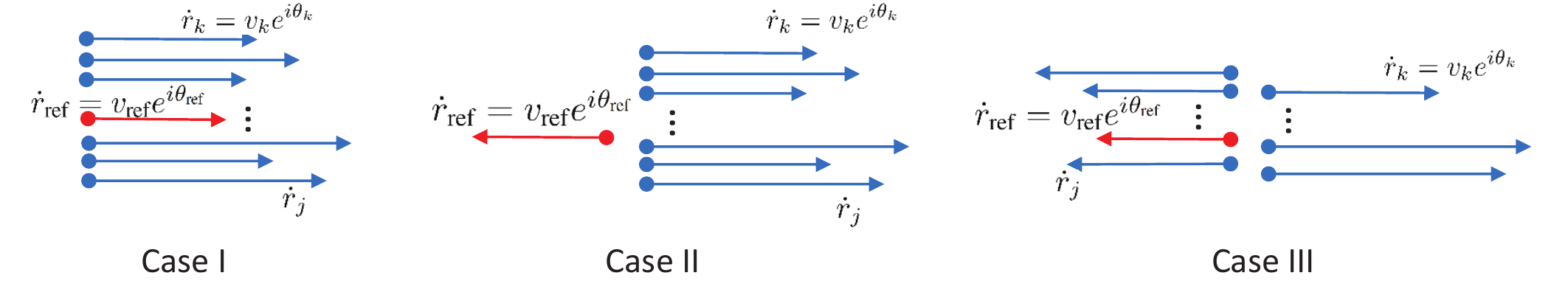}
  \caption{Illustrations of three cases of undesired equilibria under which $\|\dot { \tilde r}\| >0$ and $\text{sin}(\phi - \theta_k) = 0, \forall k$. (a) Case I: $m=0$; (b) Case II: $m=n$; (c) Case III: $1\leq m\leq n-1$. }
  \label{fig:velocity_matching}
\end{figure*}

Note that at an equilibrium point
\begin{align*}
 \left \langle \dot  {\hat r} - \dot r_{\text{ref}}, i  v_k e^{i \theta_k} \right \rangle &= \left \langle \|\dot { \tilde r}\| e^{i \phi}, i  v_k e^{i \theta_k} \right \rangle \\
 &= \|\dot { \tilde r}\| v_k   \text{sin}(\phi - \theta_k) = 0   
\end{align*}
Therefore,  at the undesired equilibria there holds $\text{sin}(\phi - \theta_k) = 0$ for all $k = 1,2,\cdots, n$, which implies that either $\theta_k  =  \phi \,\,\text{mod}\,\, 2\pi$ or $\theta_k  =  \phi+\pi \,\,\text{mod}\,\, 2\pi$.
Let $m$ be the number of vehicles with phase $(\phi+\pi \,\,\text{mod}\,\, 2\pi)$ at one of such  undesired equilibria. The three cases of the undesired equilibria are illustrated in Fig.~\ref{fig:velocity_matching}.  Now consider two extreme cases:
\begin{itemize}
\item The case that $m=0$ indicates that all the vehicles have the same phase $(\phi \,\,\text{mod}\,\, 2\pi)$. An undesired equilibrium with $m = 0$ is a local maximum of $V(\theta)$ and a small variation of any $\theta_k$ will decrease the value of $V(\theta)$. Therefore, any equilibrium with $m = 0$ and $\|\dot { \tilde r}\| >0$ is unstable.  \footnote{In the special case that all vehicles have the same constant speed $v$ that is identical to $v_{\text{ref}}$, i.e., $v_k = v = v_{\text{ref}}, \forall k$, the case of $m=0$ implies that $\dot { \tilde r} = 0$, which is a desired equilibrium and is stable.   }
\item  The case that    $m=n$ indicates that all the vehicles have the same phase $(\phi \,\,\text{mod}\,\, \pi)$. According to \eqref{eq:velocity_tracking_error} this implies that  $\dot  {\hat r} - \dot r_{\text{ref}}  =  \frac{1}{n} \sum_{k=1}^n v_k e^{i (\phi +\pi)}  - v_{\text{ref}} e^{i \phi} = (\frac{1}{n} \sum_{k=1}^n v_k + v_{\text{ref}}) e^{i (\phi +\pi)}$, which is a global maximum of $V(\theta)$ where  a small variation of any $\theta_k$ will decrease the value of $V(\theta)$. Therefore,   any equilibrium with $m = n$ is unstable. 
\end{itemize}

In the following, we will show that all other equilibria  with $1\leq m\leq n-1$ are all saddle points.
Without loss of generality, we renumber the vehicles such that the first $m$ vehicles are with phase $(\phi+\pi \,\,\text{mod}\,\, 2\pi)$. The diagonal entries of the Hessian of $V$ can be calculated as
\begin{align}
\frac{\partial ^2 V}{\partial \theta_k^2}  &= \left\{
       \begin{array}{cc}
 (1/n) v_k^2 + \|\dot { \tilde r}\| v_k, & k \in \{1, \cdots, m\}  \nonumber  \\
 (1/n) v_k^2 - \|\dot { \tilde r}\| v_k, & k \in \{m+1, \cdots, n\}
\end{array}
\right.
\end{align}
and the off-diagonal entries are
\begin{align}
\frac{\partial ^2 V}{\partial \theta_j \partial \theta_k} & = (1/n) v_j v_k \text{cos} (\theta_j - \theta_k)  \nonumber  \\
& = \left\{
       \begin{array}{cc}
 (1/n) v_j v_k,  & \,\,j, k \in \{1,  \cdots, m\},  \nonumber \\ & \text{or}  \,\,  j, k \in \{m+1,  \cdots, n\}  \nonumber  \\
 - (1/n) v_j v_k,  & \,\, \text{else}
\end{array}
\right.
\end{align}
Therefore, the Hessian can be written in a compact form
\begin{align} \label{eq:Hessian}
H_V = \frac{1}{n} v v^T + \|\dot { \tilde r}\| \text{diag}(v)
\end{align}
where the vector $v$ is defined  as $$v = [v_1,  \cdots, v_m, -v_{m+1}, \cdots, - v_n]^T.$$    
Since there exists at least one diagonal entry in the form of $(1/n) v_k^2 + \|\dot { \tilde r}\| v_k$, which is positive, the Hessian $H_V$ has at least one positive eigenvalue.  We will show the Hessian $H_V$ has at least one negative eigenvalue at any   critical point with $1\leq m\leq n-1$ by proving that there exist  vectors   $q^- \in \mathbb{R}^n$  such that    ${q^-}^T H_V q^- <0$. We first consider the case of $m = n-1$, which can happen  if and  only if the vehicle with maximum speed has phase $(\phi \,\,\text{mod}\,\, 2\pi)$, which is also the phase for the reference velocity. Define ${q^-} = [a_1,..., a_{n-1},1]^T$, where the  constant $a_i$ satisfies $0 \leq a_i <1$ for $i = 1, \cdots, n-1$ and $v^T q^- = 0$. Note that according to Proposition~\ref{ass:velocity_condition}, such non-negative $a_i$ always exists and cannot all be zero, which guarantees the existence of the vector $q^-$. Then a simple calculation yields ${q^-}^T H_V q^- = \sum_{i=1}^{n-1} a_i^2 v_i - v_n = \sum_{i=1}^{n-1} (a_i^2 v_i - a_i v_i) <0$. Therefore, an equilibrium point with $m = n-1$ is  a saddle point and is therefore unstable. 

We then consider the case of $1\leq m< n-1$.   Actually there are many options for constructing such a vector ${q^-}$. Without loss of generality, let us choose
\begin{align} \label{eq:q_vector}
q^-  = [0,  \cdots,  0, -v_n, v_{n-1}]^T
\end{align}
The existence of such $q^-$ in \eqref{eq:q_vector} is guaranteed because $m<n-1$. Note that there holds 
$v^T q^- = 0$. It then follows that
\begin{align}
{q^-}^T H_V q^-  & = \|\dot { \tilde r}\| {q^-}^T \text{diag}(v) q^- \nonumber \\
& =  \|\dot { \tilde r}\| (-v_{n-1}^2v_n - v_{n-1} v_n^2) <0
\end{align}
Hence, it is proved that such critical points with  $1\leq m< n-1$ are saddle points and therefore unstable. Consequently, all the undesired equilibria with $\left \langle \dot  {\hat r} - \dot r_{\text{ref}}, i  v_k e^{i \theta_k} \right \rangle = 0$ and $\dot  {\hat r} \neq \dot r_{\text{ref}}$ are unstable.

In summary, the above arguments have  proved that the desired equilibria at which $\dot {\hat r} = \dot r_{\text{ref}}$ are asymptotically stable, and all other equilibria  are unstable. The initial points at which the initial headings of all the vehicles are aligned are excluded in the set of initial positions because these points correspond to unstable equilibria.
\end{proof}

\begin{remark}
The result on the reference velocity tracking in Theorem~\ref{theorem:constant} is relevant and applicable to the target-tracking scenario in which  initially the target position is the same as the centroid position, and the target velocity is constant (which, as will be discussed in Section \ref{sec:spacing_tracking}, is used as the reference velocity). Extensions for time-varying velocity tracking will be discussed in the sequel. We also remark that since the group centroid is tasked to track a reference trajectory,  an average of the positions and headings of all vehicles in the group should be calculated. The proposed tracking control laws can be implemented in either
centralized or decentralized frameworks, while we  prefer a decentralized control by assuming a complete graph   in the reference velocity control to facilitate the centroid calculation. This is justified by the all-to-all communication, a commonly-used assumption in the literature on fixed-wing UAV coordination control.
The all-to-all communication scheme is a special structure of decentralized control that does not require a central station and thus features certain levels of system robustness and scalability. Though the decentralized framework may demand a higher communication overhead than a centralized scheme, it can benefit long-range flight for fixed-wing UAVs and greatly improve  the flight autonomy without interactions with any ground station. We remark that the all-to-all communication structure can be further relaxed by use of a fully distributed control with some consensus-based estimation techniques (which only requires connectivity for an  underlying communication graph). However, a fully distributed control scheme will be traded-off by additional computation and communication overhead for each vehicle, and therefore is not considered in this paper. 
\end{remark}

\begin{remark} \label{remark:real_variables}
(\textbf{Control input in real variables})  Theorem~\ref{theorem:constant} can be seen as an extension of \cite[Theorem 2]{Seyboth2014}, which considered the stabilization problem of the average linear momentum when   $\dot r_{\text{ref}} = 0$.  
The above controller involves  complex numbers and   scalar products of complex vectors. For implementation,  one can calculate the control input  \eqref{eq:controller_constant} in real variables:
\begin{align} \label{eq:real_variable}
u_k^{\text{velocity}}  &=   -  \gamma\left \langle \dot {\hat r} - \dot r_{\text{ref}}, i  v_k e^{i \theta_k} \right \rangle  \nonumber \\
&=   -  \frac{\gamma}{n}  \left \langle  \sum_{j=1}^n v_j e^{i \theta_j} , i  v_k e^{i \theta_k} \right \rangle   + \gamma \left \langle \dot r_{\text{ref}}, i  v_k e^{i \theta_k} \right \rangle \nonumber \\
&= -\frac{\gamma}{n} \sum_{j=1}^n v_k v_j \text{sin}(\theta_j - \theta_k) +  \gamma v_k v_{\text{ref}} \text{sin}(\theta_{\text{ref}} - \theta_k)
\end{align}
The phase dynamics \eqref{eq:vehicle_model_theta} with the controller \eqref{eq:real_variable} written using real variables has a similar form to the Kuramoto oscillator model, which has been studied extensively \cite{ dorfler2014synchronization, kuramoto2003chemical}, but the difference is that the speed term is involved in \eqref{eq:real_variable} for controlling non-identical constant-speed vehicles.  If we assume all vehicles have the same unit speed, the controller \eqref{eq:real_variable} then reduces to the one studied in  \cite{sepulchre2007stabilization, moshtagh2007distributed, klein2006controlled} where oscillator synchronization theory \cite{dorfler2014synchronization, kuramoto2003chemical} can apply.  
\end{remark}

\subsection{Tracking a turning reference  velocity with constant speed}
In this subsection, we consider the tracking control of a reference trajectory with constant speed and turning angular velocity. Motivated by \cite{sepulchre2008stabilization} and \cite{klein2006controlled}, the target dynamics are described by 
\begin{align} \label{eq:turning_reference}
\dot r_{\text{ref}} &= v_{\text{ref}} e^{i \theta_{\text{ref}}}  \nonumber \\
\dot \theta_{\text{ref}} & = \kappa_{\text{ref}}
\end{align}
where $v_{\text{ref}}$ is the constant speed of the reference velocity and $\kappa_{\text{ref}}$ is the angular velocity, which can be constant or non-constant and corresponds to the scaled curvature of the trajectory generated by the velocity.
In the case of tracking a time-varying reference  velocity with a constant speed, the essence of the velocity  tracking control is to design a reference velocity matching controller such that the \textit{constant reference speed} and the \textit{trajectory curvature} can be tracked. The main result is summarized in the following theorem. 

\begin{theorem}  \label{theorem:varying1}
Suppose the  turning reference velocity  and all vehicles' constant airspeeds $v_k$ satisfy the conditions in Proposition~\ref{ass:velocity_condition} and the initial headings of all the vehicles are not aligned with the initial phase of the reference velocity. Consider the following steering control law  
\begin{equation} \label{eq:controller_turning}
u_k^{\text{velocity}}  =     h_k - \gamma \left \langle \dot {\hat r} - \dot r_{\text{ref}}, i  v_k e^{i \theta_k} \right \rangle
\end{equation}
where $\hat r$ is the group centroid position, $\dot {\hat r}$ is the group centroid velocity as defined previously, $\gamma$ is a positive control gain and the $h_k$ are any real control terms that satisfy
\begin{align} \label{eq:h_k_time_varying1}
\frac{1}{n} \sum_{k=1}^n  i v_k e^{i \theta_k} h_k = i v_{\text{ref}} e^{i \theta_{\text{ref}}} \kappa_{\text{ref}}
\end{align}
Then the equilibrium point for the phase dynamics \eqref{eq:vehicle_model_theta} at which $\dot {\hat r} =  \dot r_{\text{ref}}$ is asymptotically stable and all other equilibria are unstable. The above control law will  almost globally asymptotically stabilize the group centroid velocity of the multi-vehicle group \eqref{eq:vehicle_model} to the desired reference velocity $\dot r_{\text{ref}}$ \eqref{eq:turning_reference}.   
\end{theorem}
 Theorem \ref{theorem:varying1} is a generalization of Theorem \ref{theorem:constant}. If the time-varying component $\dot \theta_{\text{ref}}=\kappa_{\text{ref}}$ of the turning reference velocity is zero, then the control \eqref{eq:controller_turning} can be reduced to the one in Theorem~\ref{theorem:constant} as one can take all $h_k$ equal to zero. 
 
\begin{proof}
The proof can be seen as an extension of the proof for Theorem \ref{theorem:constant}.  
Note that
\begin{align}
\ddot {\hat r}  & =  \frac{1}{n} \sum_{k=1}^n v_k e^{i \theta_k} i \dot \theta_k =    \frac{1}{n} \sum_{k=1}^n i v_k e^{i \theta_k}  u_k^{\text{velocity}}
\end{align}
and
\begin{align}
 \ddot r_{\text{ref}} &= i v_{\text{ref}} e^{i \theta_{\text{ref}}} \kappa_{\text{ref}}
\end{align}
Then the dynamics for the velocity tracking error can be written as
\begin{align} \label{eq:velocity_error_system1}
\ddot { \tilde r} & = \ddot {\hat r}  - \ddot r_{\text{ref}}    \nonumber  \\
& =    \frac{1}{n} \sum_{k=1}^n i v_k e^{i \theta_k}  u_k^{\text{velocity}} -  i v_{\text{ref}} e^{i \theta_{\text{ref}}} \kappa_{\text{ref}}  \nonumber  \\
& = -\frac{1}{n}\sum_{k=1}^n i v_k e^{i \theta_k}  \gamma \left \langle \dot {\hat r} - \dot r_{\text{ref}}, i  v_k e^{i \theta_k} \right \rangle
\end{align}
Consider the same Lyapunov-like  function $V = \frac{1}{2}\|\dot {\hat r} - \dot r_{\text{ref}}\|^2$  as used earlier in studying   the convergence analysis.
Note that $V(\dot { \tilde r})$ is radially unbounded and positive definite, and $V(\dot { \tilde r})$ is zero if and only if $\dot {\hat r} = \dot r_{\text{ref}}$, or equivalently $|\dot { \tilde r}| = 0$. Furthermore,   the control function \eqref{eq:controller_turning} is continuous in the time $t$ and locally Lipschitz in the state $\theta$ (uniformly in $t$). 

The  derivative of $V$ along the trajectories of the velocity tracking error system \eqref{eq:velocity_error_system1} can be calculated as
\begin{align}
\dot V & = \left \langle \dot {\hat r} - \dot r_{\text{ref}}, \ddot {\hat r}  - \ddot r_{\text{ref}}  \right \rangle \nonumber  \\
& = \left \langle \dot {\hat r} -  \dot r_{\text{ref}}, -\frac{1}{n}\sum_{k=1}^n i v_k e^{i \theta_k}  \gamma \left \langle \dot {\hat r} - \dot r_{\text{ref}}, i  v_k e^{i \theta_k} \right \rangle    \right \rangle  \nonumber  \\
& = - \frac{\gamma}{n} \sum_{k=1}^n \left \langle \dot  {\hat r} - \dot r_{\text{ref}}, i  v_k e^{i \theta_k} \right \rangle ^2 \leq 0
\end{align}
By an invariance-like principle for non-autonomous systems \cite[Theorem 8.4]{khalil2002nonlinear} and similar arguments as in the proof of Theorem \ref{theorem:constant}, it can be proven that  the solution of \eqref{eq:vehicle_model} converges to a stable equilibrium in the set ${\dot V = 0}$ with $\left \langle \dot  {\hat r} - \dot r_{\text{ref}}, i  v_k e^{i \theta_k} \right \rangle   = 0, \forall k$ and $\dot { \tilde r} = 0$, which is equivalent to $\dot {\hat r} =  \dot r_{\text{ref}}$  indicating that the time-varying reference velocity can be successfully tracked in the limit. 
\end{proof}

\begin{remark} 
The conditions in Proposition~\ref{ass:velocity_condition} guarantee the existence  of the desired equilibrium $\dot {\hat r} =  \dot r_{\text{ref}}$.
\end{remark}

We discuss how to calculate the $h_k$ below in Section \ref{sec: hk}.



\subsection{Tracking a   time-varying reference  velocity}
With the preparation of tracking controller analysis in the above subsections, in this subsection, we consider the most general case. We show the design of a general form of velocity tracking controller to regulate the group centroid velocity that aims to track a desired  time-varying reference velocity $\dot r_{\text{ref}}(t) =  v_{\text{ref}}(t) e^{i \theta_{\text{ref}}(t)}$. 
The equation of the reference velocity can be written as (see e.g., \cite{sepulchre2008stabilization} and \cite{klein2006controlled})
\begin{align} \label{eq:timevarying_velocity}
\dot r_{\text{ref}} &= v_{\text{ref}(t)} e^{i \theta_{\text{ref}}(t)}  \nonumber \\
\dot \theta_{\text{ref}} & = \kappa_{\text{ref}}(t) \nonumber \\
\dot v_{\text{ref}} &= a_{\text{ref}}(t)
\end{align}
where $\kappa_{\text{ref}}(t)$ and $a_{\text{ref}}(t)$ can be time-varying functions. To avoid cumbersome notations we will omit the argument $t$ in the following analysis. 


\begin{theorem} \label{theorem:varying2}
Suppose the time-varying reference velocity \eqref{eq:timevarying_velocity} satisfies the conditions in Proposition~\ref{ass:velocity_condition} and  the initial headings of all the vehicles are not aligned with the initial phase of the reference velocity. Consider the following steering control law
\begin{equation} \label{eq:controller_timevarying}
u_k^{\text{velocity}}  =    - \gamma \left \langle \dot {\hat r} - \dot r_{\text{ref}}, i  v_k e^{i \theta_k} \right \rangle + h_k 
\end{equation}
where $\gamma$ is a positive control gain and  the additional control terms  $h_k$  that are designed for tracking a time-varying reference velocity satisfy
\begin{align} \label{eq:h_k}
\frac{1}{n} \sum_{k=1}^n  i v_k e^{i \theta_k} h_k = i v_{\text{ref}} e^{i \theta_{\text{ref}}} \kappa_{\text{ref}} + a_{\text{ref}} e^{i \theta_{\text{ref}}}
\end{align}
Then the equilibrium set at which $\dot {\hat r}  = \dot r_{\text{ref}}$ is asymptotically stable and all other equilibria are unstable.
The above control law \eqref{eq:controller_timevarying} with the additional input \eqref{eq:h_k} will  almost globally asymptotically stabilize the group centroid velocity to the desired   reference velocity $\dot r_{\text{ref}}$. 
\end{theorem}
Theorem \ref{theorem:varying2} is a generalization of both Theorem \ref{theorem:varying1} and Theorem \ref{theorem:constant} and it treats the most general case for tracking a \textit{time-varying} reference velocity. 
If the headings of all the vehicles are not all synchronized or anti-synchronized with each other over time, then  the control terms $h_k$ are guaranteed to exist. 
We note here that there exist multiple choices for  the additional controller term $h_k$ in the right hand side of  \eqref{eq:controller_timevarying}, which will be discussed in the next subsection. 
\begin{proof}
The proof takes similar steps as that in Theorem \ref{theorem:constant}. The differences lie in the \emph{feedforward control} term $h_k$ to address the time-varying terms in the reference velocity $\dot r_{\text{ref}}$. Again, we denote the velocity tracking error as $ \dot { \tilde r}  : =  \dot  {\hat r} - \dot r_{\text{ref}}$. 
Note that
\begin{align}
\ddot {\hat r}  & =  \frac{1}{n} \sum_{k=1}^n v_k e^{i \theta_k} i \dot \theta_k =   \frac{1}{n} \sum_{k=1}^n i v_k e^{i \theta_k}  u_k^{\text{velocity}}
\end{align}
and
\begin{align}
 \ddot r_{\text{ref}} &= i v_{\text{ref}} e^{i \theta_{\text{ref}}} \kappa_{\text{ref}} + a_{\text{ref}} e^{i \theta_{\text{ref}}}
\end{align}
According to \eqref{eq:controller_timevarying},   the dynamics for the velocity tracking error can be calculated as
\begin{align} \label{eq:velocity_error_system2}
\ddot { \tilde r} & = \ddot {\hat r}  - \ddot r_{\text{ref}}    \nonumber  \\
& = -\frac{ \gamma}{n}\sum_{k=1}^n i v_k e^{i \theta_k}  \left \langle \dot {\hat r} - \dot r_{\text{ref}}, i  v_k e^{i \theta_k} \right \rangle
\end{align}
We consider the same Lyapunov function $V = \frac{1}{2}\|\dot {\hat r} - \dot r_{\text{ref}}\|^2$ which measures the difference between the current centroid velocity and the desired reference velocity. Note that $V(\dot { \tilde r}) \geq 0$ and $V(\dot { \tilde r})$ is zero if and only if $\dot {\hat r} = \dot r_{\text{ref}}$, or equivalently $|\dot { \tilde r}| = 0$.  The  derivative of $V$ along the trajectories of the velocity tracking error system \eqref{eq:velocity_error_system2} can be calculated as
\begin{align} \label{eq:dotV}
\dot V & = \left \langle \dot {\hat r} - \dot r_{\text{ref}}, \ddot {\hat r}  - \ddot r_{\text{ref}}  \right \rangle \nonumber  \\
& = \left \langle \dot {\hat r} -  \dot r_{\text{ref}}, -\frac{ \gamma}{n}\sum_{k=1}^n i v_k e^{i \theta_k}  \left \langle \dot {\hat r} - \dot r_{\text{ref}}, i  v_k e^{i \theta_k} \right \rangle    \right \rangle  \nonumber  \\
& = - \frac{\gamma}{n} \sum_{k=1}^n \left \langle \dot  {\hat r} - \dot r_{\text{ref}}, i  v_k e^{i \theta_k} \right \rangle ^2 \leq 0
\end{align}
Note that due to the feedforward control term $h_k$ in the control \eqref{eq:controller_timevarying} the dynamics $\dot \theta_k = u_k^{\text{velocity}}$ are non-autonomous. In the above analysis, one observes that the control function \eqref{eq:controller_timevarying} is continuous in the time $t$ and locally Lipschitz in the state $\theta$ (uniformly in $t$). Furthermore, the Lyapunov function $V(\dot { \tilde r}) = \frac{1}{2}\|\dot {\hat r} - \dot r_{\text{ref}}\|^2$ (which can be seen as a function of the velocity tracking error  $\dot { \tilde r} := {\hat r} - \dot r_{\text{ref}}$) is radially unbounded and positive definite for $|\dot { \tilde r}| \neq 0$, and  $\dot V \leq 0$. These conditions allow us to invoke the invariance-like principle for non-autonomous system \cite[Theorem 8.4]{khalil2002nonlinear}. Therefore, one concludes that all solutions of \eqref{eq:vehicle_model} with the controller \eqref{eq:controller_timevarying} converge to the  set $\{\dot V =0\}$ in which 
\begin{align} \label{eq:set_new_invariance}
    \left \langle \dot  {\hat r} - \dot r_{\text{ref}}, i  v_k e^{i \theta_k} \right \rangle   = 0, \,\,k = 1,2, \cdots, n
\end{align}  
Then by following  similar arguments as in the proof of Theorem \ref{theorem:constant} and by invoking  \cite[Theorem 8.5]{khalil2002nonlinear} (noting that $V$ decreases over the interval $[t, t+\delta), \forall t\geq 0$ and for some   $\delta>0$), it can be proven that  the desired equilibrium  in which $\dot { \tilde r} := {\hat r} - \dot r_{\text{ref}} = 0$ in the   set \eqref{eq:set_new_invariance} is uniformly asymptotically stable,  which indicates that the  time-varying reference velocity can be successfully tracked. Again, Proposition~\ref{ass:velocity_condition} guarantees the existence  of the desired equilibrium at which $|\dot { \tilde r}| = 0$. A similar analysis involving the Hessian matrix as in the proof of Theorem \ref{theorem:constant} 
(noting   the Lyapunov function $V(\dot { \tilde r})$ as a function of the velocity tracking error  $\dot { \tilde r}$) 
shows all other  points with  $|\dot { \tilde r}| \neq 0$ are unstable.
\end{proof}




\subsection{Discussion}
\label{sec: hk}
In this subsection, we discuss how to design the control terms $h_k$ in Theorems~\ref{theorem:varying1} and~\ref{theorem:varying2}.  
The control terms $h_k$ in \eqref{eq:h_k_time_varying1} and \eqref{eq:controller_timevarying} serve as  feedforward controls which are necessary to track the \textit{time-varying} component of the reference velocity.  

By denoting $h = [h_1, h_2, \cdots, h_n]^T$ and separating  a complex variable into real part and complex part in the form   $e^{i\theta}:= [\text{cos}(\theta), \text{sin}(\theta)]^T$,  one obtains the left-hand side of \eqref{eq:h_k_time_varying1} and \eqref{eq:h_k} in the following equivalent form:
\begin{align} \label{eq:design_hk}
&\frac{1}{n} \sum_{k=1}^n  i v_k e^{i \theta_k} h_k \nonumber \\
&= \frac{1}{n} \left[
\begin{array}{cccc}
-v_1\text{sin}(\theta_1) & -v_2\text{sin}(\theta_2) & \cdots & -v_n\text{sin}(\theta_n) \\
v_1\text{cos}(\theta_1) & v_2 \text{cos}(\theta_2) & \cdots & v_n\text{cos}(\theta_n)
\end{array} 
\right] h  \nonumber \\
& : =Ah 
\end{align}
The right-hand sides of \eqref{eq:h_k_time_varying1} and \eqref{eq:h_k} are rewritten respectively as
\begin{align}
b_1:  = i v_{\text{ref}} e^{i \theta_{\text{ref}}} \kappa_{\text{ref}}   
  =
\left[
\begin{array}{c}
 - v_{\text{ref}}\text{sin}(\theta_{\text{ref}}) \kappa_{\text{ref}}\\
 + v_{\text{ref}}\text{cos}(\theta_{\text{ref}}) \kappa_{\text{ref}}\\
\end{array}
\right]
 \end{align}
\begin{align}
b_2: &= i v_{\text{ref}} e^{i \theta_{\text{ref}}} \kappa_{\text{ref}} + a_{\text{ref}} e^{i \theta_{\text{ref}}}  \nonumber \\
& =
\left[
\begin{array}{c}
a_{\text{ref}} \text{cos}(\theta_{\text{ref}}) - v_{\text{ref}}\text{sin}(\theta_{\text{ref}}) \kappa_{\text{ref}}\\
a_{\text{ref}} \text{sin}(\theta_{\text{ref}}) + v_{\text{ref}}\text{cos}(\theta_{\text{ref}}) \kappa_{\text{ref}}\\
\end{array}
\right]
 \end{align}
Then the calculation of $h$ is equivalent  to solving a linear equation with a standard form $Ah  =   b_1$ or $Ah  =   b_2$. Since $A \in \mathbb{R}^{2 \times n}$, a sufficient condition to guarantee the existence of the  solution  is $\text{rank}(A) = 2$. For the $n$-vehicle group, the rank condition is satisfied if and only there is if at least one pair of vehicles whose headings are not aligned or anti-aligned (i.e., $\theta_i \neq \theta_j$ or $\theta_i \neq \theta_j + \pi$ for at least one pair of vehicles $i,j$). In the case of a two-vehicle group, a unique solution exists as $h = A^{-1}b$ when the two vehicles are not aligned or anti-aligned in their headings. For an $n$-vehicle group, given that the rank condition is satisfied, the solution is not unique and this provides flexibility in the controller design, while a standard least 2-norm solution could be preferred.

\begin{remark}
The reference velocity tracking in the tracking control framework is inspired by the previous papers \cite{sepulchre2007stabilization,klein2006controlled}. In contrast, here we are considering a  heterogeneous vehicle group with  constant non-identical speeds for individual vehicles, which is more general than \cite{sepulchre2007stabilization,klein2006controlled}, which discussed tracking control with unit-speed unicycle-type agents. The present results are also extensions of the tracking control strategies for  unit-speed unicycles using two special motion primitives (i.e., circular motion and parallel motion) discussed  in \cite{sepulchre2007stabilization}. A rigorous proof for the convergence of velocity tracking is presented, which is lacking in  \cite{klein2006controlled}.  Also, the non-identical speed constraints in the group present limitations for a successful tracking; note we have stated  explicit necessary conditions on the maximum reference speed and the minimum vehicle speed in Proposition~\ref{ass:velocity_condition} for the tracking controller design. 
\end{remark}

\section{Controller design phase II: Reference trajectory generation and spacing control}  \label{sec:spacing_tracking}
\subsection{Reference velocity generation for target trajectory tracking}
In the above section, we have designed reference velocity tracking controllers so that the group centroid of the fixed-wing UAVs can successfully track a reference trajectory by matching a reference velocity. In order to guarantee a successful tracking of a target, the reference trajectory should include both the target's trajectory and velocity in the construction of the reference velocity. We propose the following reference velocity
\begin{equation} \label{eq:reference_velocity_commend}
\dot r_{\text{ref}} =  \dot r_{\text{target}} + w(r_{\text{target}} - \hat r(t)) 
\end{equation}
 where $w>0$ is a weighting parameter on the relative position between the target $r_{\text{target}}$ and the group centroid $\hat r(t)$. The weighting parameter $w$ can be used to adjust the convergence speed of position tracking. A larger value of $w$ puts more weight on asymptotically tracking the position of the target, which enables a fast track to the target's trajectory. If initially the group centroid $\hat r(t)$ is collocated with the target position, then the target velocity can be used as the reference velocity. Otherwise, the relative position term $(r_{\text{target}} - \hat r(t))$ is involved in the reference velocity as a feedback term to guarantee a successful tracking to the target. The following lemma shows that by using the constructed reference velocity in \eqref{eq:reference_velocity_commend} and the velocity tracking controller in Section~\ref{sec:velocity_tracking}, the target's trajectory can be asymptotically tracked.
\begin{lemma} \label{lemma:tracking_spacing}
With the reference velocity tracking controller in \eqref{eq:controller_timevarying} designed in    Section~\ref{sec:velocity_tracking}., and the constructed reference velocity in \eqref{eq:reference_velocity_commend}, the centroid of the fixed-wing UAV group  will be asymptotically stabilized to match the target trajectory.
\end{lemma}
\begin{proof}
As proved in Theorem~\ref{theorem:varying2}, the designed reference velocity tracking controller \eqref{eq:controller_timevarying} guarantees an asymptotic convergence of the group centroid velocity to the reference velocity. Denote $\alpha(t) = \dot r_{\text{ref}} - \dot {\hat r}(t)$ as the tracking difference between the reference velocity and the centroid velocity. Then one has $\alpha(t) \rightarrow 0$ by  virtue of Theorem~\ref{theorem:varying2}. From \eqref{eq:reference_velocity_commend} we have the following equivalent equation
\begin{align}
\dot {\hat r}(t) - \dot r_{\text{target}} = -w(\hat r(t) - r_{\text{target}}) - \alpha(t)
\end{align}
Denote $\beta(t) = \hat r(t) - r_{\text{target}}$ as the trajectory tracking error between $\hat r(t)$ and $r_{\text{target}}$. Then one has $\dot \beta(t) = -w \beta(t) -\alpha(t)$. Since $w>0$ and $\alpha(t) \rightarrow 0$, one   obtains $\beta(t) \rightarrow 0$, which indicates that the trajectory tracking error converges to zero asymptotically.  
\end{proof}
In practice, one can design the weight $w$ as a distance-dependent function (i.e., $w(t) := w(\rho)$ where $\rho = \|r_{\text{ref}}(t) - r_{\text{target}}\|$) to adjust the convergence speed in different phases of the tracking process. Furthermore, since in the limit there holds $\dot  r_{\text{ref}} \rightarrow \dot r_{\text{target}}$, in order to ensure the condition in Proposition~\ref{ass:velocity_condition} is satisfied, we should also impose the  condition on the target velocity $v_{\text{min}}\geq \|v_{\text{target}}(t)\|$ to ensure a feasible tracking.

\subsection{Vehicle-target spacing control}
In this subsection, we design the spacing controller $u^{\text{spacing}}$ to ensure all vehicles stay with a bounded distance to the group centroid. As noted above, there exists a trade-off between the design of the velocity tracking controller and a spacing control. We will follow a similar idea as in \cite{kingston2007uav, klein2008coordinated} for the spacing controller design. We first present the following condition for an admissible spacing control that does not affect the performance of the centroid velocity tracking control. 


\begin{lemma} \label{lemma:condition_spacing}
Denote the spacing control vector $u^{\text{spacing}} = [u_1^{\text{spacing}}, u_2^{\text{spacing}}, \cdots, u_n^{\text{spacing}}]^T$ and define the matrix $A$ as in \eqref{eq:design_hk}. 
The additional spacing control $u^{\text{spacing}}$  satisfying
\begin{align} \label{eq:condition_spacing}
u^{\text{spacing}} \in \text{ker}(A)
\end{align}
preserves the asymptotic tracking performance of the reference velocity by the group centroid (i.e., $\dot {\hat r} \rightarrow \dot r_{\text{ref}}$ as $t\rightarrow \infty$). 
\end{lemma}
 
\begin{proof}
We  consider the same Lyapunov function as used in Section~\ref{sec:velocity_tracking}. The controller $u^{\text{velocity}}$ designed in \eqref{eq:controller_timevarying} is used here as an example for the velocity tracking analysis.
The time derivative of the Lyapunov function along the solution of the system \eqref{eq:vehicle_model} with the combined controller \eqref{eq:controller_combined} can be calculated as
\begin{align} \label{eq:lya_combine}
\dot V = & \left \langle \dot {\hat r} - \dot r_{\text{ref}}, \ddot {\hat r}  - \ddot r_{\text{ref}}  \right \rangle \nonumber  \\
 = & \left \langle \dot {\hat r} - \dot r_{\text{ref}}, \frac{1}{n}\sum_{k=1}^n i v_k e^{i \theta_k}  u_k^{\text{velocity}}   - \ddot r_{\text{ref}} \right \rangle \nonumber  \\
& + \left \langle \dot {\hat r} - \dot r_{\text{ref}}, \frac{1}{n} \underbrace{\sum_{k=1}^n i v_k e^{i \theta_k}  u_k^{\text{spacing}}}_{=0}  \right \rangle \nonumber  \\
 =& \left \langle \dot {\hat r} - \dot r_{\text{ref}}, \frac{1}{n} \sum_{k=1}^n i v_k e^{i \theta_k}  u_k^{\text{velocity}}   - \ddot r_{\text{ref}} \right \rangle  \leq 0
\end{align}
Note that the equality in the third line of the above \eqref{eq:lya_combine} is due to the condition in \eqref{eq:condition_spacing}. Therefore, $\dot V$ is invariant for any control $u^{\text{velocity}}$ of \eqref{eq:condition_spacing}, and is negative semidefinite according to the controller property. The remaining analysis is similar to the proof of previous theorems in  Section~\ref{sec:velocity_tracking} and is omitted here.
\end{proof}

For a vehicle group with $n > 2$ vehicles with non-aligned headings, the matrix $A$ is of full row rank and so has a null space of dimension $n-2$, which leaves motion freedoms for designing the spacing controller. However, it is challenging to design an admissible control input that lies in $\text{ker}(A)$ while also keeping all vehicles within a reasonable spacing around the centroid. In particular, an analytical solution for an admissible spacing control  $u^{\text{spacing}}$ is not available, while a numerical solution is usually expensive. 
Actually, even for the coordination control of \textit{identical unit-speed} vehicles, it is still an open problem to design explicit controllers to satisfy the above constraint and design requirement (see more in-depth discussions in  \cite[Chapter~2]{klein2008coordinated}).

Inspired by \cite{Pongpunwattana2007} and \cite{paley2004collective}, we  consider an alternative approach based on the  \emph{beacon control law} proposed by Paley et al. \cite{paley2004collective, paley2007oscillator} to design an intuitive spacing control. The idea is to allow each vehicle to perform limited circular trajectories around  a chosen beacon point in the reference trajectory path. In this way, the  spacing control takes a position feedback from  a reference trajectory, and is designed as
\begin{equation} \label{eq:special_spacing}
u_k^{\text{spacing}} = -\left(\omega_0 + \gamma \omega_0 \left \langle r_k - r_{\text{ref}},   v_k e^{i \theta_k} \right \rangle \right)
\end{equation}
where $\omega_o$ and $\gamma$ are positive parameters for adjusting the period and turning rate in  the circular motion: smaller 
$\omega_o$ and $\gamma$ can produce smaller magnitude for $u_k^{\text{spacing}}$ leading to smaller turning rate.
It has been proved in \cite{paley2004collective, paley2007oscillator,2017Sun_circular_formation} that the above control (in the absence of velocity tracking control) guarantees constant-speed vehicles performing circular motions around a reference point $r_{\text{ref}}$. It has also been shown in \cite{Pongpunwattana2007} by using simulation examples that the above spacing control will ensure that all the vehicles move and remain close to the centroid. We note that this control is generally not an admissible one in the null space of $A$ satisfying the condition in Lemma~\ref{lemma:condition_spacing} and therefore a perfect velocity tracking performance is not guaranteed by the addition of the spacing controller.  However, as demonstrated by numerous numerical simulations and experiments in the next sections, such a spacing control law can ensure all agents stay close to the group centroid while the group centroid position tracks the target trajectory.
\footnote{We refer the readers to \cite{sun2015collective} for several typical simulation examples that demonstrate the effectiveness and performance of the proposed tracking controllers for a group of fixed-wing UAVs with constant but non-identical airspeeds. 
In all numerical simulations, the proposed tracking controllers guarantee a desirable tracking performance through the group centroid and bounded tracking distance errors for a group of  constant-speed vehicles to collaboratively track a moving  target  (with either a constant velocity or a time-varying velocity). }


An alternative solution with a small number of UAVs is to consider projecting (\ref{eq:special_spacing}) onto the non-zero kernel of $A$. The philosophy of designing $u^{\text{spacing}}$ without affecting the centroid velocity tracking task  is inspired by \cite{klein2006controlled, antonelli2009experiments}, and is actually in the broad framework of null-space-based (NSB) robotic behavior control  proposed in \cite{antonelli2009experiments, arrichiello2010null,sadeghian2014task}. In the framework of NSB approach, each sub-task is assigned   a certain priority and the control term for a sub-task with a lower priority should  live in the null space of the control task space of those with higher priorities. However, the NSB approach is often computationally expensive and unscalable; in particular, the matrix $A$ is state-dependent, and there is no analytical formula for expressions
of its null spaces (neither by SVD nor by other projection techniques), while  even an approximation of its null space in real time
will soon become intractable when $A$ grows with sizes. Hence, this NSB approach is not suitable for real-time tracking control implemented on micro-processors in autonomous UAV systems with limited computational resources. Indeed, this approach  will quickly become infeasible  when the  number of UAVs increases. We also remark that in the two-step controller design we do not assign any priority in each sub-task.   Since an  analytical and perfect solution  for both subtasks in the target tracking control is hard to find as shown in \cite{klein2008coordinated},  one may consider ad-hoc solutions (see e.g. \cite{van2008non}) by taking into account different waypoints in the target trajectory to be used as feedback information and designing a \textit{switching} tracking controller to ensure tracking convergence and boundedness. However, in this way analytical convergence results are hard to obtain.  

\section{Experimental verification}  \label{sec:experiment}
\input{experiments.tex}

\section{Conclusion}  \label{sec:conclusion}
Fixed-wing UAVs have found increasing applications in both civilian and defense fields in recent years. Compared with rotary-wing UAVs, fixed-wing UAVs in general feature significantly larger flight ranges and longer flight duration, which are more suitable for autonomous tasks such as surveillance, circumnavigation and tracking. However, a key challenge in fixed-wing UAV coordination and control is the airspeed constraint associated with maintaining stable flight.  In this paper, we investigate  the possibility and applicability of using multiple fixed-wing UAVs with constant and possibly non-identical speeds to conduct collaborative tracking of a moving target.
 Inspired by   previous papers (e.g. \cite{sepulchre2008stabilization, klein2006controlled}), a systematic framework is proposed for the collaborative target-tracking control.  We have used the group centroid as representative of the overall UAV group in the tracking process. The design of the tracking controller consists of two parts: the reference velocity tracking control that regulates the group centroid to track a reference velocity, and a spacing controller that ensures all vehicles keep close to the group centroid. The reference velocity involves the target velocity as well as the relative position to the target as feedback to ensure the group centroid tracks the target trajectory. We have also discussed the trade-offs and limitations of using fixed-wing UAVs to track a moving target, and conditions that ensure a feasible tracking. Despite the strict constraint  of differing constant speeds, we have provided a positive solution showing  that fixed-wing UAVs are applicable for performing autonomous target tracking tasks as a group, which is  supported by both numerical simulations and real-life experiments involving a group of three fixed-wing UAVs, subject to limited wind disturbance. In the future research, we will further investigate the effect of wind flows on fixed-wing UAV coordination in general collaborative control tasks.

\section*{ACKNOWLEDGMENT}
The authors would like to thank Georg S. Seyboth for his contribution in the early stage ideas of this paper. The experimental work was supported by the \emph{Paparazzi team} in the drone lab at the Ecole Nationale de l'Aviation Civile (ENAC) in Toulouse, France. In particular, we thank Xavier Paris, Murat Bronz and Gautier Hattenberger. The authors would like to thank the reviewers and the associate editor Dr. Paley for their constructive critics.

\section*{Funding Sources}
This work is supported by the Australian Research Council Discovery Project  DP-160104500 and   DP-190100887, Data61-CSIRO, the EU H2020 Mistrale project under grant agreement no. 641606, and partially supported by the Spanish Ministry of Science and Innovation under research Grant RTI2018-098962-B-C21.

The work of Hector Garcia de Marina is supported by the grant \emph{Atraccion de Talento} with reference number 2019-T2/TIC-13503 from the Government of the Autonomous Community of Madrid.

\bibliography{Unicycle_tracking}

\end{document}

%% file: experiments.tex
The purpose of this section is to validate the proposed algorithms with fixed-wing aircraft in a series of experiments. This validation is not a mere extension of the theoretical work. In particular, we have reconsidered some assumptions in the theoretical analysis that are no longer satisfied in a real distributed control system. For example, several issues exist in practice such as the presence of delays in the transmission of information, the non-synchronization of clocks, or embedded sensors on different vehicles that are biased with respect to each other. Some of these issues may potentially have a significant impact on the performance of the overall system, especially in a decentralized control setting \cite{mou2016undirected,de2015controlling}. On top of these non-modelled effects, the fixed-wing UAV dynamics, of course, are not perfect unicycles. Therefore, one of the goals of this section is to validate the performance of the proposed algorithms in practice even when some important factors have been omitted in the development of theoretical analysis.

The experimental setup consists of one Parrot Bebop2 rotorcraft (serving as a target to be tracked) and three fixed-wing aircraft labeled as Wing 2, 3 and 4, respectively. All the vehicles are equipped with the open-source autopilot \textit{Paparazzi}, in particular, with the \emph{Apogee} autopilot\footnote{For more information, see \url{https://wiki.paparazziuav.org}, and \url{https://wiki.paparazziuav.org/wiki/Apogee/v1.00}}, which allows a rapid prototyping for distributed aerial systems as shown in our recent works \cite{2017Sun_circular_formation,2017iroscircular}. This platform also enables third parties to quickly implement and use for other purposes our proposed algorithms since it does not require any
special feature from the hardware \footnote{Firmware's source code is available at \url{https://github.com/paparazzi/paparazzi/tree/master/sw/airborne/modules/multi/ctc}, and the wiki entry for its set up \url{https://wiki.paparazziuav.org/wiki/Module/collective\_tracking\_control}}.

\begin{figure}
\centering
\includegraphics[width=1\columnwidth]{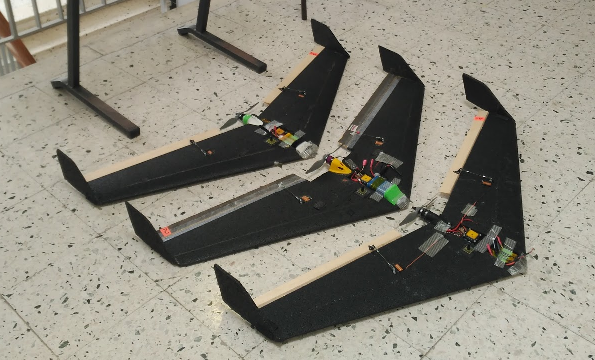}
	\caption{The three fixed-wing UAVs employed in the experiments powered by the open-source autopilot \textit{Paparazzi}.}
\label{fig: startmis}
\end{figure}

We choose the scenario of target trajectory tracking described in Section IV.A since it covers most of the presented results in this paper. In particular, we assign a rotor-craft  as an independent target that flies at the ground speed of $2$ m/s. The three aircraft fly at constant speeds between $10$ m/s and $16$ m/s while they execute onboard their control actions (\ref{eq:controller_timevarying}) with $\dot r_{\text{ref}}$ given by (\ref{eq:reference_velocity_commend}) and a positive distance-dependent weight $w(\rho) = \frac{1}{\rho}(1 - e^{-0.1\rho})$ (we remind that $\rho = \|r_{\text{ref}}(t) - r_{\text{target}}(t)\|$), and the spacing controller (\ref{eq:special_spacing}) with $\omega_0 = 0.25$ rads/sec. The chosen gains for (\ref{eq:controller_timevarying}) and (\ref{eq:reference_velocity_commend}) were $0.001$ for both cases. The relative positions between the vehicles are calculated onboard by having the UAVs broadcasting their absolute positions obtained by a GPS. {We employ \emph{X-Bee} modules, which create a network among them. This network supports air-to-air communication without ground intervention, and therefore a ground central station is not required. This feature is already implemented and available in the \emph{Paparazzi} project}. The broadcasting frequency for the target is 5Hz, whereas for the aircraft it is 10Hz. Indeed, the fact that there exist communication losses (i.e., communication dropouts over short intervals) between aircraft, and that the aircraft process the information at different times, leads to possible discrepancies among the three vehicles about the relative positions.


\begin{figure}
\centering
\includegraphics[width=1\columnwidth]{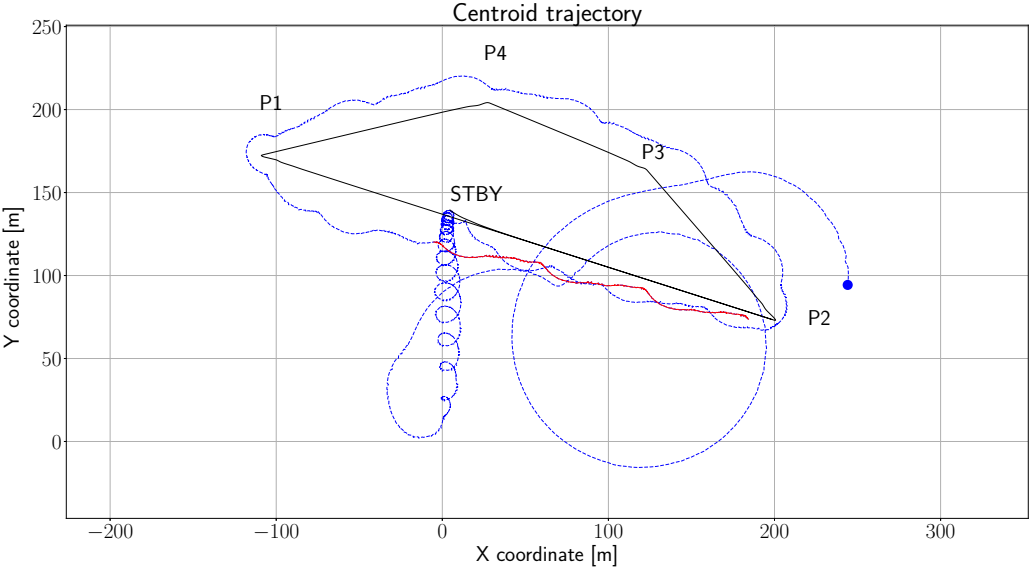}
	\caption{Tracking performance of the centroid trajectory. The blue  trajectory shows the  centroid trajectory, where the dot denotes   the starting point. The red trajectory shows the last two minutes of the mission. The black trajectory is the target trajectory once it starts moving from STBY.}
\label{fig: startcen}
\end{figure}

\begin{figure}
\centering
\includegraphics[width=1\columnwidth]{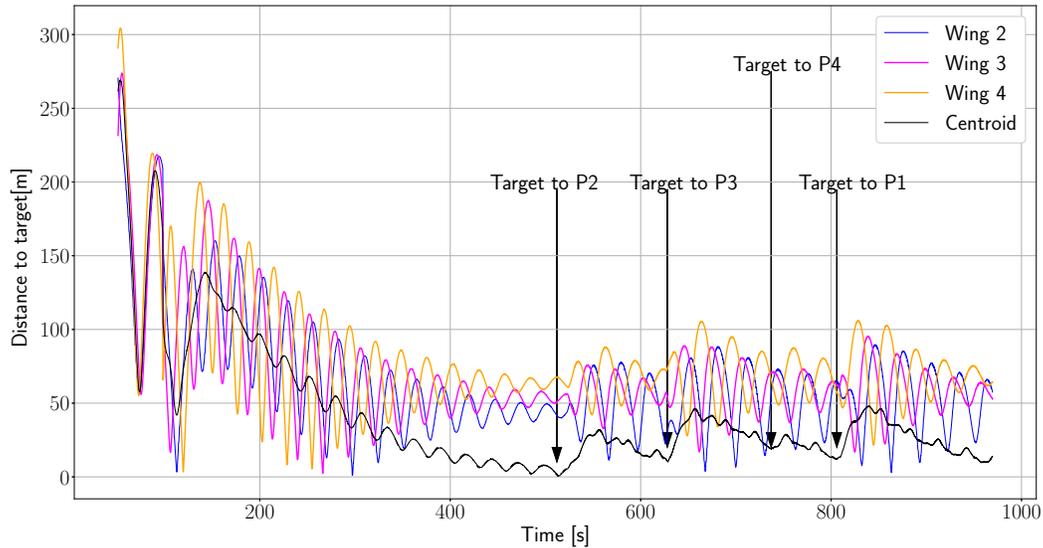}
	\caption{Time evolution of the relevant distances during the experiment. Before the target moves to P2, it remains stopped at the STBY waypoint.}
\label{fig: distances}
\end{figure}

\begin{figure*}
\centering
\begin{subfigure}{.48\columnwidth}
  \centering
  \includegraphics[width=\linewidth]{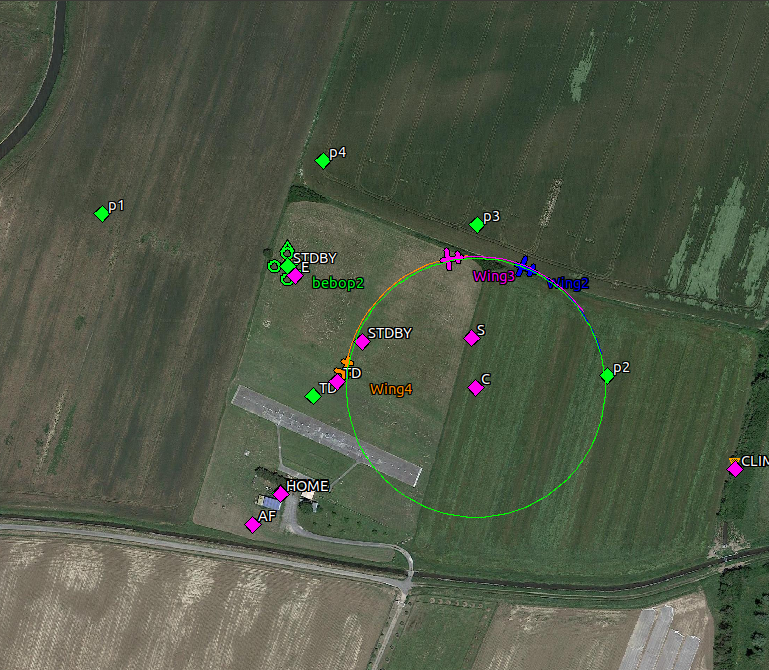}
  \caption{t = $90$secs}
\end{subfigure}
\begin{subfigure}{.50\columnwidth}
  \centering
  \includegraphics[width=\linewidth]{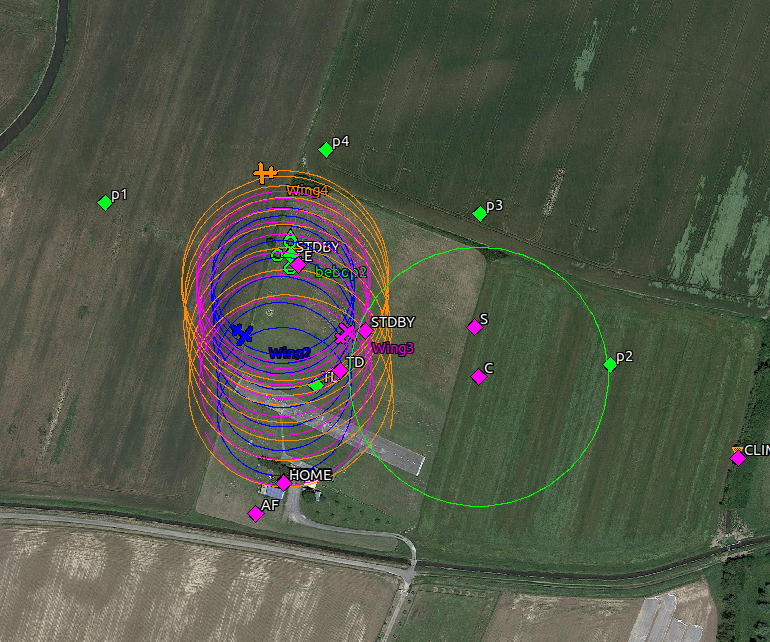}
   \caption{t = $420$secs}
\end{subfigure}
\begin{subfigure}{.49\columnwidth}
  \centering
  \includegraphics[width=\linewidth]{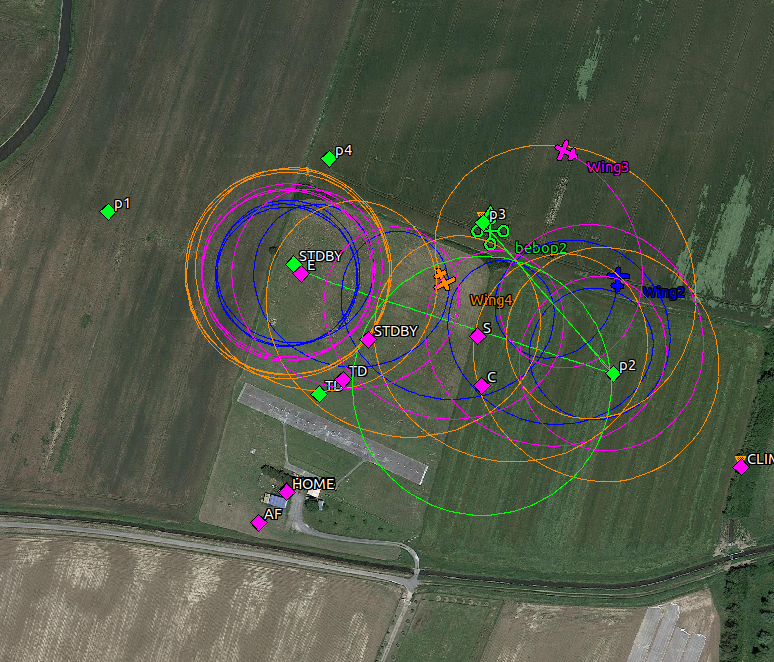}
   \caption{t = $685$secs}
\end{subfigure}
\begin{subfigure}{.49\columnwidth}
  \centering
  \includegraphics[width=\linewidth]{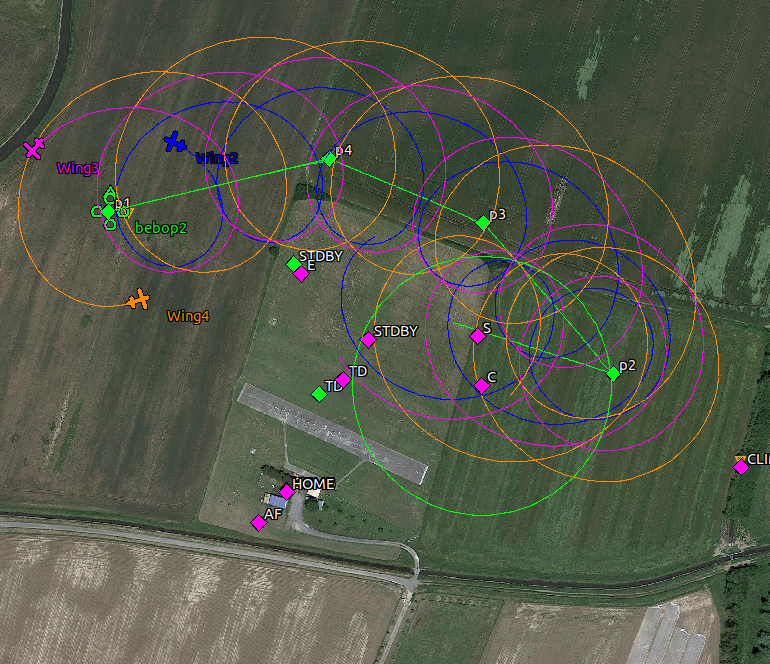}
	\caption{t = $809$secs} 
\end{subfigure}%
	\caption{Screenshots from the \emph{Paparazzi} ground control station during the mission. Note how in (a) and (b) the target is stopped at STBY, and afterward it starts moving. It can be noted how the aircraft enclose the target while the centroid of the team is close to it.}
\label{fig: gs}
\end{figure*}

{The experiment was performed in a radio-control club in Muret, a city close to Toulouse in France. The wind velocity had an average speed of 3 m/s and an almost  constant direction. Note that this wind speed has a noticeable impact on the ground speed of the aircraft. Nevertheless, as the experimental results indicate, such a wind speed does not have an impact on the intended performance of the algorithm. We consider $\theta$ as the heading angle (vector velocity), and not the attitude \emph{yaw} angle. If there is no wind, both angles are the same in our setup. In practice, when we consider the heading instead of the yaw for the unicycle model, the aircraft ends up compensating the lateral wind by \emph{crabbing} so that aerodynamic angle \emph{sideslip} is almost zero\footnote{Crabbing happens when the inertial velocity makes an angle with the nose heading due to wind. \emph{Slipping} happens when the aerodynamic velocity vector makes an angle (sideslip) with the body ZX plane. Slipping is (almost) always undesirable, because it degrades aerodynamic performance. Crabbing is not an issue for the aircraft.}.}

We divide the mission into two stages. In the first stage, we place the rotorcraft on the ground at the fixed point STBY. Then we launch the three aircraft to orbit around the waypoint C. We denote in green color the starting circular trajectory in Figure \ref{fig: gs}a. At time $t=90$ secs, the three aircraft start the algorithm from their \emph{stand-by circle}. In Figure \ref{fig: startcen}, we plot the position of the centroid of the team in blue color, where the blue dot (at the right side on the plot) is its initial position at $t=90$ secs. With our setup and initial conditions, it takes around seven minutes until the centroid of the team converges to the fixed starting position of the target (see Figures \ref{fig: distances} and \ref{fig: gs}b). At the end of this first stage, as shown in Figure \ref{fig: gs}b, the STBY point (or green rotorcraft) is at the barycenter of the triangle described by the three aircraft.

In the second stage, after the convergence to the target's starting point, the target starts flying following a closed path with a speed of $2$m/s. The target's trajectory can be checked in Figure \ref{fig: startcen} in black color, or the green straight lines in Figure \ref{fig: gs}. The target changes its velocity (but not its speed) every time it travels to a different waypoint. Then, the algorithm successfully drives the centroid of the team to converge again to the target's moving position as illustrated in Figure \ref{fig: distances}. This fact can be noticed by checking that the rotorcraft in Figure \ref{fig: gs}d is at the barycenter of the triangle described by the three aircraft.

During the experiments, the rotorcraft broadcasts its position and velocity, since we assumed that the target cooperates as   might be the case in, for example, a sports event (e.g., the cyclist tracking problem illustrated in Fig.~\ref{fig: tour}). Nevertheless, such broadcasting can be replaced by onboard sensing to collect the relative position of the aircraft with respect to the target. Indeed, the estimation of the target's velocity might be done numerically from the sensed relative positions. Also, it can be more sophisticated with the addition of estimators onboard driven by the error signal between the group's centroid and the target. In fact, in case of a wrong reading of the target's velocity, such an error signal would converge to a non-zero constant. We expect that small estimation errors can be tolerated. In fact, with the spacing control law, a perfect tracking is often not possible due to the constant-speed constraint, while one can only expect an approximate yet still satisfactory tracking in practice.

\begin{remark}
Since we are employing an \emph{all-to-all communication} network, one would be interested in assessing the scalability of the network in terms of communication budget or needed bandwidth by the vehicles. In our experiments, the vehicles share four floating numbers (two floats for their 2D positions, and two floats for their 2D velocity) at $10$Hz frequency, which is $1.280$kbps. If there were 100 flying aircraft, an individual vehicle would need \emph{to attend} $128$kbps. The employed  radio modem (each with 25 USD dollars) in our experiments can handle up to $250$kbps.
\end{remark}

\begin{remark}Indeed, we consider that there is no limit in the aircraft's turning rate in our proposed control actions. However,  the heading control law  establishes bounded turning rates that depend on the initial conditions of the system, and on the chosen $\omega_0$ and the gain $\gamma$ in the different proposed control actions. In practice, one can always adjust the values of the two key parameters $\omega_0$ and   $\gamma$
to tune the magnitude of the heading control input $u_k$ so as to change the turning rate that meets practical
requirements on turning radius.  We provide a quick simulator\footnote{\url{https://github.com/paparazzi/paparazzi/blob/master/sw/ground_segment/python/multi/collective_tracking_control/ctc_simulation.py}} linked to the autopilot realistic values in order to assess the order of magnitude of the demanded turning rate during a whole mission. Therefore, one can see beforehand an educated guess of the demanded turning rates for all autonomous UAVs in a  collective tracking operation.
\end{remark}